\documentclass[conference]{IEEEtran}
\IEEEoverridecommandlockouts
\usepackage{cite}
\usepackage{amsmath,amssymb,amsfonts}
\usepackage[ruled,linesnumbered,vlined]{algorithm2e}
\usepackage{graphicx}
\usepackage{textcomp}
\usepackage{xcolor}
\usepackage{enumerate}
\usepackage{ntheorem}
\usepackage{subfigure}
\usepackage{balance}
\usepackage{hyperref}

\newcommand{\nop}[1]{}

\newtheorem{definition}{Definition}
\newtheorem{example}{Example}
\newtheorem{lemma}{Lemma}
\newtheorem*{proof}{Proof:}
\def\BibTeX{{\rm B\kern-.05em{\sc i\kern-.025em b}\kern-.08em
    T\kern-.1667em\lower.7ex\hbox{E}\kern-.125emX}}
\begin{document}

\makeatletter
\newcommand{\linebreakand}{%
  \end{@IEEEauthorhalign}
  \hfill\mbox{}\par
  \mbox{}\hfill\begin{@IEEEauthorhalign}
}
\makeatother

\title{Top-$L$ Most Influential Community Detection Over Social Networks (Technical Report)}






\author{\IEEEauthorblockN{Nan Zhang}
\IEEEauthorblockA{\textit{East China Normal University} \\
Shanghai, China \\
51255902058@stu.ecnu.edu.cn}
\and
\IEEEauthorblockN{Yutong Ye}
\IEEEauthorblockA{\textit{East China Normal University} \\
Shanghai, China \\
52205902007@stu.ecnu.edu.cn}
\and
\IEEEauthorblockN{Xiang Lian}
\IEEEauthorblockA{\textit{Kent State University} \\
Kent, United States \\
xlian@kent.edu}
\and
\IEEEauthorblockN{Mingsong Chen}
\IEEEauthorblockA{\textit{East China Normal University} \\
Shanghai, China \\
mschen@sei.ecnu.edu.cn}
}

\maketitle

\begin{abstract}

In many real-world applications such as social network analysis and online marketing/advertising, the \textit{community detection} is a fundamental task to identify communities (subgraphs) in social networks with high structural cohesiveness. While previous works focus on detecting communities alone, they do not consider the collective influences of users in these communities on other user nodes in social networks. Inspired by this, in this paper, we investigate the influence propagation from some \textit{seed communities} and their influential effects that result in the \textit{influenced communities}. We propose a novel problem, named \textit{\underline{Top-$L$} most \underline{I}nfluential \underline{C}ommunity \underline{DE}tection} (Top$L$-ICDE) over social networks, which aims to retrieve top-$L$ seed communities with the highest influences, having high structural cohesiveness, and containing user-specified query keywords. In order to efficiently tackle the Top$L$-ICDE problem, we design effective pruning strategies to filter out false alarms of seed communities and propose an effective index mechanism to facilitate efficient Top-$L$ community retrieval. We develop an efficient Top$L$-ICDE answering algorithm by traversing the index and applying our proposed pruning strategies. We also formulate and tackle a variant of Top$L$-ICDE, named \textit{diversified top-$L$ most influential community detection} (DTop$L$-ICDE), which returns a set of $L$ diversified communities with the highest diversity score (i.e., collaborative influences by $L$ communities). Through extensive experiments, we verify the efficiency and effectiveness of our proposed Top$L$-ICDE and DTop$L$-ICDE approaches over both real-world and synthetic social networks under various parameter settings.



\end{abstract}

\begin{IEEEkeywords}
Top-$L$ Most Influential Community Detection, Diversified Top-$L$ Most Influential Community Detection 
\end{IEEEkeywords}

\section{Introduction}
Recently, the \textit{community detection} (CD) has gained significant attention as a fundamental task in various real-world applications, such as online marketing/advertising \cite{chen2010scalable,tang2015influence,tu2022viral}, friend recommendation \cite{song2022friend}, and social network analysis \cite{fan2018octopus}. Many previous works \cite{wang2020efficient,liu2020truss,sun2020index,liu2021efficient} usually focused on identifying communities only (i.e., subgraphs) with high structural cohesiveness in social networks. However, they overlooked the collective influences that communities may exert on other users (e.g., family members, or friends) within social networks, which play a significant role in Word-Of-Mouth effects.



In this paper, we will formulate and tackle a novel problem called \textit{top-$L$ most influential community detection over social networks} (Top$L$-ICDE). This Top$L$-ICDE problem aims to detect top-$L$ communities of people from social networks who have specific interests (e.g., sports, movies, traveling), not only with close social relationships (i.e., forming dense subgraphs with highly connected users) but also with high impacts/influences on other users in social networks. 

Below, we give a motivating example of our Top$L$-ICDE problem in real applications of online advertising/marketing.\vspace{-1.5ex}


\setlength{\textfloatsep}{0pt}
\begin{figure}[t]
    \vspace{-5ex}
    \centering
    \subfigure[social network $G$]{
        \includegraphics[height=3.5cm]{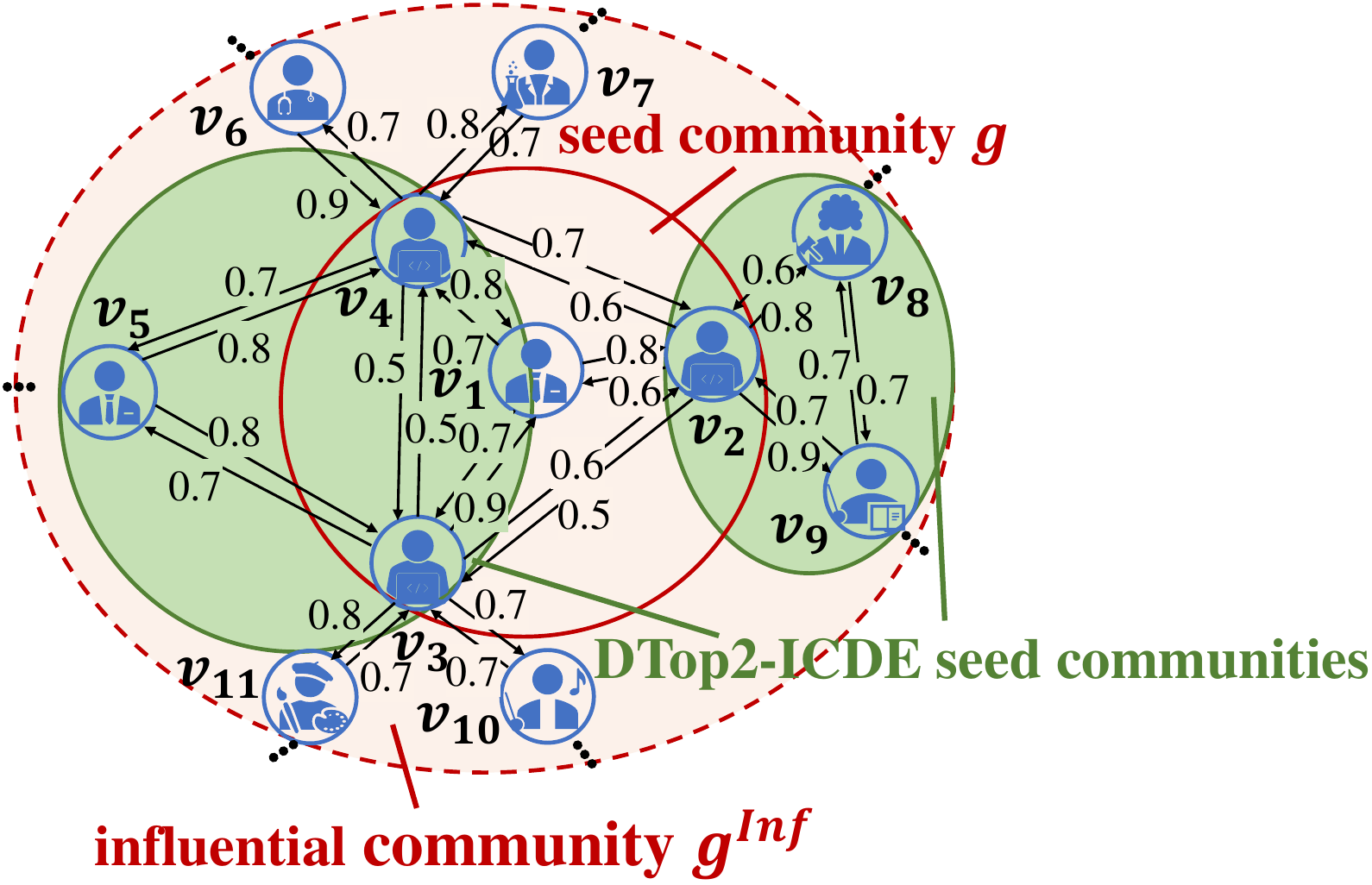}
        \label{subfig:example_graph}
    }
    \hspace{-0.4cm}
    \subfigure[keyword sets of vertices]{
        \includegraphics[height=3.5cm]{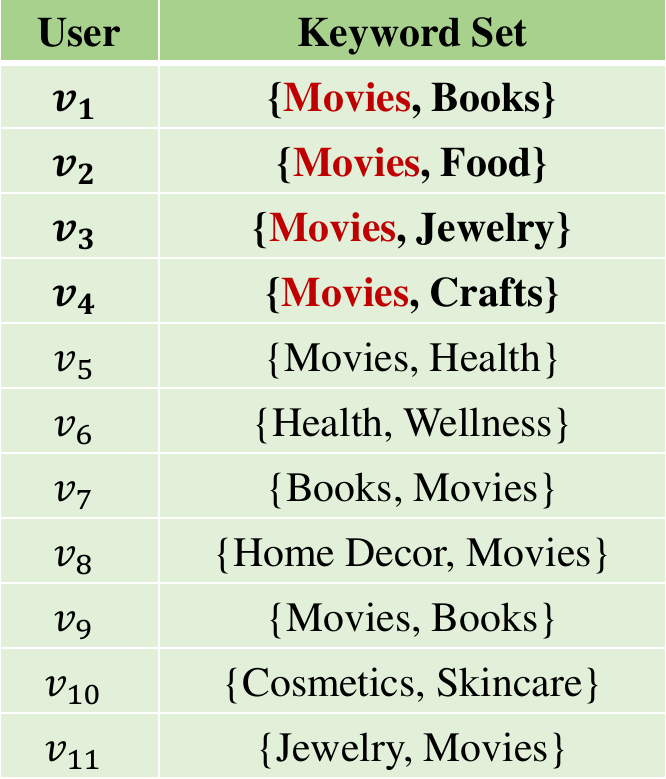}
        \label{subfig:example_table}
    }
    \caption{An example of the Top$L$-ICDE problem over social network.}
    \label{fig:motivation}
\end{figure}

\begin{example}\textbf{(Online Advertising and Marketing Over Social Networks)} 
\label{example:topl_icde_example}
In real applications of online advertising and marketing, a sales manager wants to find several communities of users on social networks who might be interested in buying a new product. 

Intuitively, people in the same communities who know each other should share a sufficient number of common friends within communities, so that they are more likely to purchase the product together via group buying (e.g., on Groupon \cite{wang2017game}).

What is more, after such community users receive group buying coupons/discounts, the sales manager wants these users to maximally influence other users on social networks via their posts/recommendations. In this case, the sales manager needs to perform the Top$L$-ICDE operator to find some seed communities with the highest influences on social networks (also considering those influenced users). 

Figure \ref{fig:motivation} illustrates an example of a social network, $G$, which contains user vertices like $v_1 \sim v_5$, and directed edges (e.g., edge $e_{v_1,v_2}$) representing the relationships (e.g., friend, colleague, or family) between two users (e.g., $v_1$ and $v_2$). As illustrated in Figure \ref{subfig:example_graph}, each edge has an influence weight between two vertices, for example, directed edge $e_{v_1,v_2}$ has the weight $0.8$, indicating the influence from user $v_1$ to user $v_2$. Moreover, as depicted in Figure \ref{subfig:example_table}, each user is associated with a set of keywords that represent one's favorite product categories (e.g., $\{Movies, Books\}$ for user $v_1$). 

To achieve good online advertising/marketing effects, one potential seed community, $g$, is shown in Figure \ref{subfig:example_graph} (i.e., inner solid circle),  where users in $g$ are all interested in ``Movies'', with high structural connectivity (i.e., any two users in $g$ are friends and also share two other common friends), and have the highest influences on other user nodes in a larger influenced community $g^\mathit{Inf}$ (i.e., outer dashed ellipse, as shown in Figure \ref{subfig:example_graph}). \qquad $\blacksquare$




\end{example}\vspace{-1ex}

Inspired by the example above, in this paper, we consider the Top$L$-ICDE problem, which obtains $L$ groups of highly connected users (called \textit{seed communities}) containing some query keywords (e.g., shopping preferences or preferred products) and with the highest influences (ranks) on other users. The resulting $L$ communities are potential customer groups for effective advertising/marketing with high influences to promote new products in social networks.

Formally, we will first define a seed community $g$ in social networks as a structurally dense subgraph (under the semantics of $(k, r)$-truss \cite{huang2017AttributeDrivenCommunitySearch,al2020topic}), where any two connected users in $g$ must have at least $(k-2)$ common friends (i.e., be contained in $\geq (k-2)$ triangle motifs, indicating stable user relationships), users in $g$ are close to a center vertex (i.e., within $r$ hops), and each individual vertex in $g$ contains at least one query keyword.
Then, we will consider the influenced communities under an influence propagation model from seed communities, which are ranked by influential scores of seed communities. Our Top$L$-ICDE problem aims to retrieve top-$L$ seed communities with the highest ranks.

\nop{\color{blue} However, in real-world social network, the user groups influenced by the communities with higher influence may overlap, which lower the online advertising/marketing effects. Therefore, another valuable problem is obtaining $L$ seed communities with the highest influences and the most influenced users to enhance the promoting effect, which is called \textit{diversified top-$L$ most influential community detection over social networks} (DTop$L$-ICDE).}

Due to the large scale of social networks and the complexities of retrieving communities under various constraints, it is rather challenging to efficiently and effectively tackle the Top$L$-ICDE problem over big social networks. Furthermore, since many parameters like query keywords or thresholds (e.g., influential score, radius, and support) are online specified by users, it is not clear how to enable efficient Top$L$-ICDE processing with ad-hoc constraints (or query predicates).

In order to efficiently tackle the Top$L$-ICDE problem, in this paper, we propose a two-phase general framework for Top$L$-ICDE processing, which consists of offline pre-computation and online Top$L$-ICDE processing phases. In particular, we will present effective pruning mechanisms (w.r.t query keywords, radius, edge support, and influential scores) to safely filter out false alarms of candidate seed communities and reduce the Top$L$-ICDE problem search space. Moreover, we will design an effective indexing mechanism to facilitate our proposed pruning strategies and develop an efficient Top$L$-ICDE processing algorithm via index traversal.

Furthermore, we also consider a variant of Top$L$-ICDE, namely \textit{diversified top-$L$ most influential community detection over social networks} (DTop$L$-ICDE), which retrieves a set of $L$ diversified communities with the highest diversity score. Different from Top$L$-ICDE that returns $L$ individual communities each of which can be a candidate community for online marketing/advertising), DTop$L$-ICDE obtains one set of $L$ diversified communities that collaboratively influence other users with the highest diversity score (i.e., considering the overlaps of their influenced communities). 

In Example~\ref{example:topl_icde_example} (Figure~\ref{fig:motivation}), if a sales manager wants to enhance the promotion effect of advertising to $L$ seed communities, we need to consider those users influenced by more than one seed community. Since users usually buy the product once only, we would like to reduce the overlaps of $L$ influenced communities and maximize the influential effect (i.e., the diversity score). In this case, we can issue the DTop$L$-ICDE variant to obtain $L$ diversified communities. In Figure \ref{subfig:example_graph}, for $L=2$, subgraphs $\{v_3, v_4, v_5\}$ and $\{v_2, v_8, v_9\}$ are two DTop$2$-ICDE seed communities (with only one influenced user $v_1$).

Based on the Top$L$-ICDE algorithm, we prove that the DTop$L$-ICDE problem is NP-hard, and develop an approximation greedy algorithm to process the DTop$L$-ICDE query efficiently.

We make the following major contributions in this paper:
\begin{enumerate}
    \item We formally define the problem of the top-$L$ most influential community detection over social networks (Top$L$-ICDE) and its variant DTop$L$-ICDE in Section~\ref{sec:problem_definition}.
    \item We design a two-phase framework to efficiently tackle the Top$L$-ICDE problem in Section~\ref{sec:solution_framework}.
    \item We propose effective pruning strategies to reduce the Top$L$-ICDE search space in Section~\ref{sec:pruning_strategies}.
    \item We devise offline pre-computation and indexing mechanisms in Section~\ref{sec:offline_process} to facilitate effective pruning and an efficient Top$L$-ICDE algorithm to retrieve community answers in Section~\ref{sec:online_topl_icde_process}.
    \item We prove that DTop$L$-ICDE is NP-hard, and develop an efficient DTop$L$-ICDE processing algorithm to retrieve diversified community answers in Section~\ref{sec:online_dtopl_icde_process}.
    \item We demonstrate through extensive experiments the efficiency and effectiveness of our Top$L$-ICDE processing approach over real/synthetic graphs in Section~\ref{sec:experiments}.
\end{enumerate}

Section~\ref{sec:related_work} reviews previous works on community search/detection, influence maximization, and influential / diversified community. Finally, Section~\ref{sec:conclusion} concludes this paper.

\section{Problem Definition}
\label{sec:problem_definition}
Section~\ref{sec:social_network} formally defines the data model for social networks. Section~\ref{subsec:information_propagation_model} gives the definition of the information propagation model over social networks. Finally, Section~\ref{subsec:top_L_most_influential_community_detection} provides the definition of our problems.

\subsection{Social Networks}
\label{sec:social_network}

First, we model a social network by an attributed, undirected, and weighted graph as follows.

\begin{definition}
\label{def:social_network}
(\textbf{Social Network, $G$}) A social network $G$ is a connected graph represented by a triple $(V(G), E(G),\Phi(G))$, where $V(G)$ and $E(G)$ are the sets of vertices and edges in $G$, respectively, and $\Phi(G)$ is a mapping function: $V(G) \times V(G) \rightarrow E(G)$. Each vertex $v_i$ has a keyword set $v_i.W$, and each edge $e_{u,v} \in E(G)$ is associated with a weight $p_{u,v}$, which indicates the probability that user $u$ activates user $v$.  
\end{definition}

In Definition \ref{def:social_network}, the keyword set $v_i.W$, e.g., $\{movies, sports, \cdots\}$, denotes the topics that user $v_i$ is interested in.

\subsection{Information Propagation Model}
\label{subsec:information_propagation_model}

To describe the influence spread over social networks $G$, we utilize the \textit{maximum influence arborescence} (MIA) model \cite{chen2010ScalableInfluenceMaximization}. Given a specific path from $u$ to $v$ (i.e., a non-cyclic user sequence), denoted as $P_{u,v} = \left\langle u=u_1, u_2,\ldots, u_m=v \right\rangle$,  the \textit{propagation probability}, ${pp}(P_{u,v})$, of path $P_{u,v}$ is given by:
\begin{equation}
\label{eq:user_to_user_path_influence}
    {pp}(P_{u,v}) = \prod_{i=1}^{m-1}{p_{u_i,u_{i+1}}}.
\end{equation}
where $p_{u_i,u_{i+1}}$ is the weight of edge $e_{u_i,u_{i+1}}$ on path $P_{u,v}$.

In Eq.~(\ref{eq:user_to_user_path_influence}), we give the probability of the influence propagation between two user vertices via a specific path $P_{u,v}$ in $G$. In reality, there are multiple possible paths from $u$ to $v$. Thus, the MIA model uses the \textit{maximum influence path} (MIP) \cite{chen2010ScalableInfluenceMaximization}, denoted as $\mathit{MIP}_{u, v}$, to evaluate the influence propagation from $u$ to $v$. The MIP is defined as a path with the highest propagation probability below:
\begin{equation}
    \label{eq:maximum_influence_path}
    \mathit{MIP}_{u, v}=\mathop{argmax}\limits_{P_{u,v}}{pp(P_{u,v})}.
\end{equation}

This way, the \textit{user-to-user propagation probability}, ${upp}(u,v)$, for all paths between users $u$ and $v$ is defined by:
\begin{equation}
    \label{eq:user_to_user_influence}
    {upp}(u,v) = {pp}(\mathit{MIP}_{u, v}).
\end{equation}

The problem of computing the influence spread $\sigma(S)$ is known to be NP-hard \cite{chen2010ScalableInfluenceMaximization}, and existing algorithms \cite{kempe2003MaximizingSpreadInfluence, feige1998ThresholdLnApproximating} can achieve an approximation factor of $(1-{1/e} +{\varepsilon})$, where $e$ is the natural constant and $\varepsilon>0$.

\subsection{
Our Top$L$-ICDE Problem}
\label{subsec:top_L_most_influential_community_detection}

\noindent {\bf Seed Community.} We first give the definition of the \textit{seed community} in social networks $G$ below. 

\begin{definition}
(\textbf{Seed Community, $g$}) Given a social network $G$, a center vertex $v_q$, an integer support $k$, the maximum radius, $r$, of seed communities, and a set, $Q$, of query keywords, a \textit{seed community}, $g$, is a connected subgraph of $G$ (denoted as $g \subseteq G$), such that:
\begin{itemize}
\item $v_q \in V(g)$;
\item for any vertex $v_l \in V(g)$, we have $dist(v_q, v_l) \leq r$;
\item $g$ is a $k$-truss \cite{cohen2008TrussesCohesiveSubgraphs}, and;
\item for any vertex $v_l$$\in$$V(g)$, its keyword set $v_l.W$ must contain at least one query keyword in $Q$ (i.e., $v_l.W \cap Q \neq \emptyset$),
\end{itemize}
where $dist(x, y)$ is the shortest path distance between $x$ and $y$ in subgraph $g$.
\label{def:seed_community}
\end{definition}

The seed community $g$ (given in Definition~\ref{def:seed_community}) is a connected subgraph that is  1) centered at $v_q$, 2) with a maximum radius $r$, 3) being a $k$-truss, and, 4) with each vertex $v_l$ containing at least one query keyword in $Q$. Here, $g$ is a $k$-truss subgraph \cite{cohen2008TrussesCohesiveSubgraphs}, if any edge in $g$ is contained in at least $(k-2)$ triangles. 

In real applications of online marketing and advertising, the seed community usually contains strongly connected users, who are provided with coupons/discounts of products to maximally influence other users.

\noindent {\bf Influenced Community.}
The \textit{influenced community}, $g^\mathit{Inf}$, is a subgraph influenced by a seed community $g$.
Based on the MIA model, we define the \textit{community-to-user propagation probability}, $cpp(g, v)$, from seed community $g$ to a vertex $v$ as follows.
\begin{eqnarray}
\label{eq:community_to_user_influence}
cpp (g, v) = 
    \begin{cases}
    \max_{\forall u\in V(g)}\{upp(u,v)\}, & v \notin V(g); \\
    1, & v\in V(g).
    \end{cases}
\label{eq:cpp}
\end{eqnarray}

We use $cpp(g, v)$ to compute the maximum possible influence (e.g., posts, tweets) from one of the users in the seed community $g$ to user $v$ through some paths. Then, we provide the definition of the \textit{influenced community} $g^\mathit{Inf}$ below.

\begin{definition}
\label{def:influenced_community}
(\textbf{Influenced Community, $g^\mathit{Inf}$}) Given a social network $G$, a seed community $g$, and an influence threshold $\theta$ ($\in [0, 1)$), the influenced community,  $g^\mathit{Inf}$, of $g$ is a subgraph of $G$, where each vertex $v$ in $V(g^\mathit{Inf})$ satisfies the condition that $cpp (g, v) \geq \theta$.
\end{definition}

\noindent {\bf The Influential Score of the Influenced Community.} To evaluate the propagation effect from a seed community $g$ to its influenced community $g^\mathit{Inf}$, we give the definition of the \textit{influential score}, $\sigma(g)$, for the influenced community $g^\mathit{Inf}$: 

\begin{equation}
\label{eq:influence_score}
    \sigma(g)=\sum_{v \in V(g^\mathit{Inf})}{cpp (g, v)}.
\end{equation}

The influential score $\sigma(g)$ (given in Eq.~(\ref{eq:influence_score})) sums up all the community-to-user propagation probabilities $cpp (g, v)$ from $g$ to vertices $v$ in the influenced community $g^\mathit{Inf}$. Intuitively, high influential score $\sigma(g)$ indicates that the seed community $g$ may influence either a few users $v$ with high community-to-user propagation probabilities $cpp(g, v)$, or many users $v$ even with low $cpp(g, v)$ values.


\noindent {\bf The Problem of Top-$L$ Most Influential Community Detection Over
Social Networks.} We are now ready to define the problem of detecting top-$L$ communities with the highest influences.

\begin{definition}
\label{def:topl_icde_problem}
(\textbf{Top-$L$ Most Influential Community Detection Over Social Networks, Top$L$-ICDE})
Given a social network $G$, a positive integer $L$, a threshold $\theta$, a support, $k$, of the trusses, the maximum radius, $r$, of seed communities, and a set, $Q$, of query keywords, the problem of \textit{top-$L$ most influential community detection over
social networks} (Top$L$-ICDE) retrieves $L$ seed communities $g_i$, such that:
\begin{itemize}
\item $g_i$ satisfy the constraints of seed communities (as given in Definition \ref{def:seed_community}), and;
\item these $L$ seed communities $g_i$ have the influenced communities, $g_i^\mathit{Inf}$, with the highest influential scores $\sigma(g_i)$,
\end{itemize}
where $\sigma(g_i)$ is given by Eq.~(\ref{eq:influence_score}). 
\end{definition}

\noindent {\bf A Variant of Top$L$-ICDE (Diversified Top-$L$ Most Influential Community Detection Over Social Networks).} In Definition~\ref{def:topl_icde_problem}, the Top$L$-ICDE problem returns $L$ individual seed communities with the highest influences. Note that, these $L$ individual communities may influence the same users (i.e., with high overlaps of the influenced users). In order to achieve higher user impacts, in this paper, we also consider a variant of Top$L$-ICDE, namely \textit{diversified top-$L$ most influential community detection over social networks} (DTop$L$-ICDE), which obtains a set of $L$ diversified communities with the highest collaborative influences on other users. 

Given a set, $S$, of communities, we formally define its diversity score, $D(S)$, to evaluate the collective influence of communities in $S$ on other users:
\begin{equation}
\label{eq:diversity_score}
    D(S)=\sum_{\forall v \in V(G)}{\max_{\forall g \in S}\{cpp (g, v)}\}.
\end{equation}

The diversity score in Eq.~(\ref{eq:diversity_score}) sums up the maximum possible community-to-user propagation probabilities, $cpp(g, v)$, from any community $g$ in $S$ to the influenced users $v$. Intuitively, a higher diversity score indicates higher impacts from communities in $S$.

For simplicity, we use $\Delta D_{g_i}(S)$ to represent the increment of the diversity score for adding the subgraph $g_i$ to set $S$, i.e., $\Delta D_{g_i}(S) = D(S\cup\{g_i\}) - D(S)$.

Next, we define our DTop$L$-ICDE problem which returns $L$ diversified seed communities with the highest diversity score (i.e., collaborative influences on other users).

\begin{definition} (\textbf{Diversified Top-$L$ Most Influential Community Detection Over Social Networks, DTop$L$-ICDE})
\label{def:dtopl_icde_problem}
Given a social network $G$, a positive integer $L$, a threshold $\theta$, a support, $k$, of the trusses, the maximum radius, $r$, of seed communities, and a set, $Q$, of query keywords, the problem of \textit{diversified top-$L$ most influential community detection over
social networks} (DTop$L$-ICDE) retrieves a set, $S$ of $L$ seed communities $g_i$, such that:
\begin{itemize}
\item $g_i$ satisfy the constraints of seed communities (as given in Definition~\ref{def:seed_community}), and;
\item the set $S$ of $L$ seed communities $g_i$ has the highest diversity score $D(S)$ (as given by Eq.~(\ref{eq:diversity_score})).
\end{itemize}
\end{definition}

Intuitively, the DTop$L$-ICDE problem (given in Definition \ref{def:dtopl_icde_problem}) finds a set, $S$, of $L$ communities that have the highest collaborative influence, that is, the diversity score $D(S)$ in Eq.~(\ref{eq:diversity_score}), which is defined as the summed influence from communities $g$ in $S$ to the influenced users.

\noindent {\bf Challenges.} A straightforward method to tackle the Top$L$-ICDE problem is to first obtain all possible seed communities (subgraphs) of the data graph $G$, then check the constraints of these seed communities, and finally rank these seed communities based on their influential scores. Similarly, for DTop$L$-ICDE, we can also compute the diversity score for any combination of $L$ communities, and choose a set of size $L$ with the highest diversity score. However, such straightforward methods are quite inefficient, especially for large-scale social networks, due to the high costs of the constraint checking over an exponential number of possible seed communities (or community combinations), as well as the costly computation of influence/diversity scores (as given in Eqs.~(\ref{eq:influence_score}) and ~(\ref{eq:diversity_score})). Thus, the processing of the Top$L$-ICDE problem (and its variant DTop$L$-ICDE) raises up the efficiency and scalability issues for detecting (diversified) top-$L$ most influential communities over large-scale social networks.

\begin{table}[t!]
\caption{Symbols and Descriptions}\scriptsize
\vspace{-0.1in}
\label{table:symbols_and_descriptions}
\begin{center}
\begin{tabular}{|l|p{6.1cm}|}
\hline
\textbf{Symbol}&{\textbf{Description}} \\
\hline\hline
$G$ & a social network\\
\hline
$V(G)$ & a set of $n$ vertices $v_i$\\
\hline
$E(G)$ & a set of edges $e(u,v)$\\
\hline
$g$ (or $g_i$) & a seed community (subgraph)\\
\hline
$p_{u,v}$& a propagation probability that user $u$ activates its neighbor $v$\\
\hline
$g^\mathit{Inf}$ (or $g^\mathit{Inf}_i$) & the influenced community of a seed community $g$ (or $g_i$) \\
\hline
$hop(v_i, r)$ & a subgraph centered at vertex $v_i$ and with radius $r$ \\
\hline
$v_i.W$ & a set of keywords that user $v_i$ is interested in \\
\hline
$v_i.BV$ & a bit vector with the hashed keywords in $v_i.W$ \\
\hline
$\sigma(g)$ & the influential score of the influenced community $g^\mathit{Inf}$\\
\hline
$Q$ & a set of query keywords\\
\hline
$Q.BV$ & a bit vector with the hashed query keywords in $Q$ \\
\hline
$k$ & the support threshold in $k$-truss for seed communities\\
\hline
$sup(e_{u,v})$ & the support of edge $e_{u,v}$\\
\hline
$r_{max}$ & the maximum possible radius of seed communities\\
\hline
$r$ & the user-specified radius of seed communities\\
\hline
$\theta$ & the influence threshold\\
\hline
\end{tabular}
\label{tab1}
\end{center}
\end{table}

\section{Our Top$L$-ICDE Processing Framework}
\label{sec:solution_framework}
Algorithm \ref{algo:the_solution_framework} presents our 
framework to efficiently answer Top$L$-ICDE queries, 
which consists of two phases, i.e., \textit{offline pre-computation} and \textit{online Top$L$-ICDE processing phases}.

In the first offline pre-computation phase, we pre-process the social network graph by pre-computing data (e.g., influential score bounds) to facilitate online query answering and constructing an index over these pre-computed data. In particular, for each vertex $v_i$ in data graph $G$, we first hash its associated keyword set $v_i.W$ into a bit vector $v_i.BV$ (lines 1-2). Then, for each $r$-radius subgraph $hop(v_i, r)$ centered at vertex $v_i$ and with radius $r$ ($\in [1, r_{max}]$), we offline pre-compute data (e.g., support/influence bounds) to facilitate the pruning (lines 3-5). Next, we construct a tree index $\mathcal{I}$over the pre-computed data (line 6).

In the online Top$L$-ICDE processing phase, for each query, we traverse the index $\mathcal{I}$ to efficiently retrieve candidate seed communities, by integrating our proposed effective pruning strategies (w.r.t., keyword, support, radius, and influential score) (lines 7-8). Finally, we refine these candidate seed communities by computing their actual influential scores and return top-$L$ most influential communities with the highest influential scores (line 9). 

\noindent {\bf Discussions on the DTop$L$-ICDE Framework.} We will discuss the DTop$L$-ICDE framework later in Section \ref{sec:online_dtopl_icde_process}, which follows the Top$L$-ICDE framework but applies specifically designed pruning/refinement techniques to retrieve $L$ diversified communities (i.e., DTop$L$-ICDE query answers).


\begin{algorithm}[!ht]
\caption{{\bf The Top$L$-ICDE Process Framework}}\small
\label{algo:the_solution_framework}
\KwIn{
    \romannumeral1) a social network $G$;
    \romannumeral2) a set, $Q$, of query keywords;    
    \romannumeral3) the support, $k$, of the truss for each seed community;
    \romannumeral4) the maximum radius, $r$, of seed communities;
    \romannumeral5) the influence threshold $\theta$, and;
    \romannumeral6) integer parameter $L$

}
\KwOut{
    a set, $S$, of top-$L$ seed communities
}

\tcp{\bf offline pre-computation phase}
\For{each $v_i \in V(G)$}{
    hash keywords in $v_i.W$ into a bit vector $v_i.BV$\;
    \For{$r = 1$ \KwTo $r_{max}$}{
        extract $r$-hop subgraph $hop(v_i, r)$ of vertex $v_i$\;
        offline pre-compute data, $v_i.R$, w.r.t. the support upper bound $ub\_sup(\cdot)$ and influence upper bound $\mathit{Inf}_{ub}$ for subgraph $hop(v_i, r)$\;
    }
}
build a tree index $\mathcal{I}$ over graph $G$ with pre-computed data as aggregates\;

\tcp{\bf online Top$L$-ICDE processing phase}
\For{each Top$L$-ICDE query}{
    traverse the tree index $\mathcal{I}$ by applying \textit{keyword}, \textit{support}, \textit{radius}, and \textit{influential score pruning}  strategies to retrieve candidate seed communities\;
    refine candidate seed communities to obtain top-$L$ seed communities with the highest influential scores\;    
}
\end{algorithm}

\section{Pruning Strategies}
\label{sec:pruning_strategies}
In this section, we present effective pruning strategies to reduce the Top$L$-ICDE problem search space in our framework (line 8 of Algorithm \ref{algo:the_solution_framework}). 

\subsection{Keyword Pruning}
In this subsection, we first provide an effective \textit{keyword pruning} method. From Definition \ref{def:seed_community}, any vertex in the seed community $g$ must contain at least one query keyword in $Q$. Thus, our keyword pruning method filters out those candidate subgraphs $g$ containing some vertices without query keywords.

\begin{lemma}
\label{lemma:keyword_pruning}
{\bf (Keyword Pruning)} Given a set, $Q$, of query keywords and a candidate subgraph $g$, subgraph $g$ can be safely pruned, if there exists at least one vertex $v_i \in V(g)$ such that: $v_i.W \cap Q = \emptyset$ holds, where $v_i.W$ is the keyword set associated with vertex $v_i$.
\end{lemma}
\begin{proof}
\label{proof:keyword_pruning_lemma}
Intuitively, if $v_i.W \cap Q = \emptyset$ holds for any vertex in $g$, it means that user $v_i$ in the social network is not interested in any keyword in the query keyword set. It does not satisfy the Definition~\ref{def:seed_community} so that subgraph $g$ can be safely pruned, which completes the proof. \qquad $\square$
\end{proof}
\nop{
\begin{proof}
Please refer to Appendix \ref{proof:keyword_pruning_lemma} for the detailed proof.
\end{proof}
}

\subsection{Support Pruning}
According to Definition \ref{def:seed_community}, the seed community $g$ should be a $k$-truss \cite{cohen2008TrussesCohesiveSubgraphs}, that is, the support $sup(e_{u, v})$ of each edge $e_{u,v} \in E(g)$ (defined as the number of triangles that contain edge $e_{u,v}$ in the seed community $g$) must be at least $(k-2)$.

Assume that we can offline obtain an upper bound, $ub\_sup(e_{u, v})$, of the support $sup(e_{u, v})$ on edge $e_{u,v}$ in $g$. Then, we have the following lemma to discard those candidate seed communities $g$ containing some edges with low support.
\begin{lemma}
\label{lemma:support_pruning}
    {\bf (Support Pruning)} Given a seed community $g$ and a parameter $k$, subgraph $g$ can be safely pruned if there exists 
    an edge $e_{u,v}\in E(g)$ satisfying 
    $ub\_sup(e_{u,v}) < (k-2)$.
\end{lemma}
\begin{proof}
\label{proof:support_pruning_lemma}
For a subgraph $g$, we can calculate the support $sup(e_i)$ for each $e_i \in E(g)$. In the definition of k-truss\cite{cohen2008TrussesCohesiveSubgraphs}, the value of $sup(e_i)$ is the number of triangles made by $e_i$ and other pairs of edges. Moreover, in a $k$-truss subgraph, each edge $e$ must be reinforced by at least $k - 2$ pairs of edges, making a triangle with edge $e$. According to our definition of seed community (Definition~\ref{def:seed_community}), subgraph $g$ can be safely pruned by a specific $k$, which completes the proof. \qquad $\square$
\end{proof}
\nop{
\begin{proof}
Please refer to Appendix \ref{proof:support_pruning_lemma} for the detailed proof.
\end{proof}
}

\noindent {\bf Discussions on How to Obtain the Support Upper Bound $\mathbf{ub\_sup(e_{u,v})}$.} To enable the support pruning, we need to calculate the support upper bound $ub\_sup(e_{u,v})$ of edge $e_{u,v}$ in a seed community $g$. Since the seed community $g$ is a subgraph of the data graph $G$, the support of edge $e_{u,v}$ in $g$ is thus smaller than or equal to that in $G$ (in other words, the number of triangles containing $e_{u,v}$ in $g$ is less than or equal to that in $G$). Therefore, we can use the edge support in the data graph $G$ (or any supergraph of $g$) as the support upper bound $ub\_sup(e_{u,v})$ of edge $e_{u,v}$.

\subsection{Radius Pruning}
From Definition~\ref{def:topl_icde_problem}, the maximum radii, $r$, of seed communities are online specified by users, which limits the shortest path distance between the center vertex and any other vertices to be not more than $r$. We provide the following pruning lemma with respect to radius $r$.

\begin{lemma}
\label{lemma:radius_pruning}
{\bf (Radius Pruning)} Given a subgraph $g$ (centered at vertex $v_i$) and the maximum radius, $r$, of seed communities, subgraph $g$ can be safely pruned, if there exists a vertex $v_l\in V(g)$ such that $dist(v_i, v_l) > r$, where function $dist(x, y)$ outputs the number of hops between vertices $x$ and $y$ in $g$.
\end{lemma}
\begin{proof}
\label{proof:radius_pruning_lemma}
According to the second bullet in Definition~\ref{def:seed_community}, a subgraph with maximum radius $r$ means that there exists a center vertex whose \textit{the shortest path distance}($dist(.,.)$) to any other vertex is less than $r$. For a subgraph $g$ and vertex $v_i \in g$, if there exists vertex $v_j \in g$ whose $dist(v_i, v_j)$ is greater than $r$, it does not satisfy the definition of seed community (Definition~\ref{def:seed_community}) so that subgraph $g$ can be safely pruned, which completes the proof. \qquad $\square$
\end{proof}

\nop{
\begin{proof}
Please refer to Appendix \ref{proof:radius_pruning_lemma} for the detailed proof.
\end{proof}
}

If a subgraph $g$ (centered at vertex $v_i$) has the radius greater than $r$, then it does not satisfy the radius constraint of the seed community.


This radius pruning method can enable the offline pre-computation by extracting subgraphs for any possible radius $r\in [1, r_{max}]$, that is, $hop(v_i, r)$. In other words, those vertices with distance to vertex $v_i$ greater than radius $r$ can be safely ignored, so we only need to focus on $r$-hop subgraphs $hop(v_i, r)$.

\subsection{Influential Score Pruning}
\label{sec:influential_score_pruning}
Next, we discuss the \textit{influential score pruning} method below, which filters out those seed communities with low influential scores.


\begin{lemma}
\label{lemma:influential_score_pruning}
{\bf (Influential Score Pruning)} Assume that we have obtained $L$ seed communities $g_i$ so far, and let $\sigma_L$ be the smallest influential score among these $L$ seed communities $g_i$. Any subgraph $g$ can be safely pruned, if it holds that $ub\_\sigma(g) \leq \sigma_L$, where $ub\_\sigma(g)$ is the upper bound of the influential score $\sigma(g)$. 
\end{lemma}
\begin{proof}
\label{proof:influence_upperbound_pruning_lemma}
Given the smallest influential score among $L$ seed communities $g_i$, $\sigma_L$, and a seed community $g$, the influential score of $g$, $\sigma(g)$, and the upper bound of $\sigma(g)$, $ub\_\sigma(g)$. It holds that $ub\_\sigma(g) > \sigma(g)$ in accordance with the definition of upper bound. If $\sigma_L > ub\_\sigma(g)$, there is $\sigma_L > ub\_\sigma(g) > \sigma(g_i)$. So $g$ can not be added to the top $L$ set and can be safely pruned, which completes the proof. \qquad $\square$
\end{proof}

\nop{
\begin{proof}
Please refer to Appendix \ref{proof:influence_upperbound_pruning_lemma} for the detailed proof.
\end{proof} 
}


\noindent {\bf Discussions on How to Obtain the Upper Bound $ub\_\sigma(g)$ of the Influential Score.} Based on Eq.~(\ref{eq:influence_score}), the influential score $\sigma(g)$ of seed community $g$ is given by summing up the community-to-user propagation probabilities $cpp(g, v)$, for $v\in g^\mathit{Inf}$, where $cpp(g, v) \geq \theta$. Since threshold $\theta$ is online specified by the user, we can offline pre-select $m$ thresholds $\theta_1$, $\theta_2$, ..., and $\theta_m$ (assuming $\theta_1 < \theta_2< ... <\theta_m$), and pre-calculate the influential scores $\sigma_z(g)$ w.r.t. thresholds $\theta_z$ (for $1\leq z\leq m$). 
Given an online threshold $\theta$, if $\theta \in [\theta_z, \theta_{z+1})$ holds, we will use $\sigma_z(g)$ as the influential score upper bound $ub\_\sigma(g)$, where $\sigma_z(g)$ is the influential score of $g$ using threshold $\theta_z$.

\section{Offline Pre-Computation and Indexing}
\label{sec:offline_process}
In this section, we discuss how to offline pre-compute data for social networks to facilitate effective pruning in Section \ref{subsec:pre-computation}, and construct indexes over these pre-computed data to help with online Top$L$-ICDE processing in Section \ref{subsec:indexing_mechanism}.

\begin{algorithm}[!ht]
\caption{{\bf Offline Pre-Computation}}\small
\label{algo:offline_precomputation}
\KwIn{
    \romannumeral1) a social network $G$;
    \romannumeral2) the maximum value of radius $r_{max}$;
    \romannumeral3) $m$ influence thresholds \{$\theta_1, \theta_2,$ $\ldots,$ $\theta_m$\};
}
\KwOut{
    pre-computed data $v_i.R$ for each vertex $v_i$;
}

\For{each $v_i \in V(G)$}{
    \tcp{keyword bit vectors}
    hash all keywords in the keyword set $v_i.W$ into a bit vector $v_i.BV$\;
    $v_i.R = \{v_i.BV\}$\;
    \tcp{edge support upper bounds}
    \For{each $e_{u,v} \in E(hop(v_i, r_{max}))$}{
        compute edge support upper bounds $ub\_sup(e_{u,v})$ w.r.t. $hop(v_i, r_{max})$\;
    }
} 

\For{each $v_i \in V(G)$}{
    \For{each $r=1$ \KwTo $r_{max}$}{
        $v_i.BV_r = \bigvee_{\forall v_l\in V(hop(v_i, r))} v_l.BV$\;
        $v_i.ub\_sup_r = \max_{\forall e_{u,v} \in E(hop(v_i, r))} ub\_sup(e_{u,v})$\;
        \tcp{influential score upper bounds}
        \For{each $\theta_z \in \{\theta_1, \theta_2, \ldots, \theta_m\}$}{
            $\sigma_z(hop(v_i,r)) = {\sf calculate\_influence}(hop(v_i,r), \theta_z)$\;
            add $(\sigma_z(hop(v_i,r)), \theta_z)$ to $v_i.R$\;
        }
        add $v_i.BV_r$ and $v_i.ub\_sup_r$ to $v_i.R$\;
    }
}
\Return $v_i.R$\;
\end{algorithm}

\subsection{Offline Pre-Computed Data for Top$L$-ICDE Processing}
\label{subsec:pre-computation}

In order to facilitate online Top$L$-ICDE processing, we perform offline pre-computations over data graph $G$ and obtain some aggregate information of potential seed communities, which can be used for our proposed pruning strategies to reduce the problem search space.  
Specifically, for each vertex $v_i\in V(G)$, we first hash the keyword set $v_i.W$ into a bit vector $v_i.BV$ of size $B$. Each edge $e_{u,v}$ is associated with the edge support upper bound $ub\_sup(e_{u,v})$ in $r_{max}$-hop subgraph $hop(v_i, r_{max})$ (i.e., a pair of $(v_i, ub\_sup(e_{u,v})$).
Next, we use the radius pruning (as given in Lemma \ref{lemma:radius_pruning}) to enable offline pre-computations of $r$-hop subgraphs. In particular, starting from each vertex $v_i$, we traverse the data graph $G$ in a breadth-first manner (i.e., BFS), and obtain $r$-hop subgraphs, $hop(v_i, r)$, centered at $v_i$ with radii $r\in [1, r_{max}]$. For each $r$-hop subgraph, we calculate and store aggregates in a list $v_i.R$, in the form $(v_i.BV_r, v_i.ub\_sup_r, [(\sigma_z, \theta_z)])$ as follows:

\begin{itemize}
    \item a bit vector, $v_i.BV_r$, which is obtained by hashing all keywords in the keyword sets $v_l.W$ of vertices $v_l$ in the $r$-hop subgraph $hop(v_i, r)$ into a position in the bit vector (i.e., $v_i.BV_r = \bigvee_{\forall v_l\in V(hop(v_i, r))} v_l.BV$);
    \item an upper bound, $v_i.ub\_sup_r$, of all support bounds $ub\_sup(e_{u,v})$ for edges $e_{u,v}$ in the $r$-hop subgraph $hop(v_i, r)$ (i.e., $v_i.ub\_sup_r = \max_{\forall e_{u,v} \in E(hop(v_i, r))} ub\_sup(e_{u,v})$), and;
    \item $m$ pairs of influential score upper bounds and influence thresholds $(\sigma_z(hop(v_i, r)), \theta_z)$ (for $1\leq z\leq m$).
\end{itemize}

More details of the aggregates are provided below:

\noindent{\bf The Computation of Keyword Bit Vectors $v_i.BV_r$:} We first obtain the keyword bit vector $v_i.BV$ of size $B$ for each vertex $v_i\in V(G)$, and then compute the one, $v_i.BV_r$, for $r$-hop subgraph $hop(v_i, r)$.
Specifically, for each vertex $v_i$, we first initialize all bits in a vector $v_i.BV$ with zeros. Then, for each keyword $w$ in the keyword set $v_i.W$, we use a hashing function  $f(w)$ that maps a keyword $w$ to an integer between $[0, B-1]$ and set the $f(w)$-th bit position to 1 (i.e., $v_i.BV[f(w)] = 1$).
Next, for all vertices $v_l$ in $r$-hop subgraph $hop(v_i, r)$, we perform a bit-OR operator over their bit vectors $v_l.BV$. That is, we have $v_i.BV_r = \bigvee_{\forall v_l\in V(hop(v_i, r))} v_l.BV$.

\noindent{\bf The Computation of Support Upper Bounds $v_i.ub\_sup_r$:} For each radius $r\in [1, r_{max}]$, we compute a support upper bound $v_i.ub\_sup_r$ as follows. We first calculate the support upper bound, $ub\_sup(e_{u,v})$, for each edge $e_{u,v}$ in the $r_{max}$-hop subgraph $hop(v_i, r_{max})$. Then, the maximum $ub\_sup(e_{u,v})$ among all edges in $hop(v_i, r_{max})$ is selected as the support upper bound $v_i.ub\_sup_r$.

\noindent{\bf The Computation of Influential Score Upper Bounds $\sigma_z(hop(v_i, r))$ (w.r.t. $\theta_z$):} Since seed communities $g$ are subgraphs of some $r$-hop subgraphs $hop(v_i, r)$, as given in Eq.~(\ref{eq:influence_score}), the influential score $\sigma_z(hop(v_i, r))$ (w.r.t. influence threshold $\theta_z$) must be greater than or equal to $\sigma_z(g)$ and is thus an influential score upper bound. In other words, we overestimate the influence of a seed community $g$ inside $r$-hop subgraph $hop(v_i, r)$, by assuming that $g=hop(v_i, r)$.

To calculate the influential score $\sigma_z(hop(v_i, r))$, we first start from each $r$-hop subgraph $hop(v_i, r)$ ($=g$) and then expand $hop(v_i, r)$ to obtain its influenced community $g^\mathit{Inf}$ via a graph traversal. A vertex $u$ is included in $g^\mathit{Inf}$, if it holds that $cpp(g, u) \geq \theta_z$, where $cpp(g, u)$ is given by Eq.~(\ref{eq:cpp}). The graph traversal algorithm terminates, when $cpp(g, u) < \theta_z$ holds. Finally, we use Eq.~(\ref{eq:influence_score}) to calculate the influential score $\sigma_z(hop(v_i, r))$ of the expanded graph (w.r.t., $\theta_z$).

\noindent{\bf Offline Computation Algorithm:} Algorithm \ref{algo:offline_precomputation} illustrates the epseudo code of offline data pre-computation in the data graph $G$ that can facilitate online Top$L$-ICDE processing. In particular, for each vertex $v_i \in V(G)$, all the keywords of the keyword set $v_i.W$ are hashed into a bit vector $v_i.BV$ and stored in a list $v_i.R$ (lines 1-3). Then, we compute the support for each edge $e_{u,v}$ in $hop(v_i,r_{max})$ and use the maximum one as the support upper bound $ub\_sup(e_{u,v})$ for the edge $e_{u,v}$ (lines 4-5). 
Next, for each vertex $v_i$ and pre-selected radius $r$ ranging from $1$ to $r_{max}$, we calculate the pre-computed data for subgraph $hop(v_i, r)$, including keyword bit vector $v_i.BV_r$ (lines 6-8), edge support upper bound $v_i.ub\_sup_r$ (line 9), and the upper bound of influential score $\sigma_z(hop(v_i, r))$ w.r.t. $\theta_z$ (lines 10-11). All these pre-computed data are added to the list $v_i.R$ (lines 12-13), which is returned as the output of the algorithm (line 14).

\noindent{\bf Complexity Analysis:}
In Algorithm~\ref{algo:offline_precomputation}, for each $v_i \in V(G)$ in the first loop, the time complexity of computing a keyword bit vector $v_i.BV$ is given by $O(|W|)$ (lines 2-3). Let $avg\_deg$ denote the average number of vertex degrees. Since there are ${avg\_deg}^{r_{max}}$ edges in $r_{max}$-hop and the cost of the support upper bound computation is a constant (counting the common neighbors), the time complexity of obtaining edge support upper bounds is given by $O({avg\_deg}^{r_{max}})$ (lines 4-5).

In the second loop of each $v_i \in V(G)$, there are ${avg\_deg}^{r-1}$ vertices in the $r$-hop w.r.t. $r$. Thus, for each $r\in[1,r_{max}]$, the time complexity of computing $v_i.BV_r$ and $v_i.ub\_sup_r$ is given by $O(B{avg\_deg}^{r-1})$ and $O({avg\_deg}^{r-1})$, respectively (lines 8-9). Let $g_r^\mathit{inf}$ as the subgraph containing the influenced users from $hop(v_i,r)$. As described in Section \ref{subsec:online_topl_icde_algorithm}, the time complexity of {\sf calculate\_influence($\cdot$,$\cdot$)} is given by $O((|E(g_r^\mathit{inf})|+|V(g_r^\mathit{inf})|)log|V(g_r^\mathit{inf})|)$, the time complexity of obtaining influential score upper bounds is given by $O(m \cdot ((|E(g_r^\mathit{inf})|+|V(g_r^\mathit{inf})|)log|V(g_r^\mathit{inf})|))$ (lines 10-13).

Therefore, the time complexity of the total offline pre-computation is $O(|V(G)| \cdot (|W|+{avg\_deg}^{r_{max}} + r_{max}\cdot((B+1){avg\_deg}^{r-1} + m \cdot ((|E(g_r^\mathit{inf})|+|V(g_r^\mathit{inf})|)log|V(g_r^\mathit{inf})|))))$.

On the other hand, the space complexity of the two loops is $O(|V(G)|+|E(G)|)$ and $O(|V(G)|\cdot r_{max}\cdot (B+2m+1))$, respectively. Thus, the total space complexity of the offline pre-computation is given by $O(|V(G)|+|E(G)|+|V(G)|\cdot r_{max}\cdot (B+2m+1))$.

\subsection{Indexing Mechanism}
\label{subsec:indexing_mechanism}

In this subsection, we illustrate the details of building a tree index, $\mathcal{I}$, over pre-computed data of social networks $G$, which can be used for performing online Top$L$-ICDE processing. 

\noindent{\bf The Data Structure of Index $\mathcal{I}$:} We will construct a hierarchical tree index, $\mathcal{I}$, over social networks $G$, where each index node $N$ contains multiple entries $N_i$, each corresponding to a subgraph of $G$. 

Specifically, the tree index $\mathcal{I}$ contains two types of nodes, leaf and non-leaf nodes.

\noindent\underline{\it Leaf Nodes:} Each leaf node $N$ in index $\mathcal{I}$ contains multiple vertices $v_i \in V(G)$. Each vertex $v_i$ is associated with the following pre-computed data in $v_i.R$ (w.r.t. each possible radius $r \in [1, r_{max}]$).
\begin{itemize}
    \item a keyword bit vector $v_i.BV_r$;
    \item edge support upper bound $v_i.ub\_sup_r$, and;
    \item $m$ pairs of influential score upper bounds and influence thresholds ($\sigma_z(hop(v_i, r))$, $\theta_z$).
\end{itemize}

\noindent\underline{\it Non-Leaf Nodes:} Each non-leaf node $N$ in index $\mathcal{I}$ has multiple index entries, each of which, $N_i$, is associated with the following aggregate data below (w.r.t. each radius $r \in [1, r_{max}]$).
\begin{itemize}
    \item an aggregated keyword bit vector $N_i.BV_r = \bigvee_{\forall v_l\in N_i} v_l.BV_r$; 
    \item the maximum edge support upper bound $N_i.ub\_sup_r = \max_{\forall v_l\in N_i} v_l.ub\_sup_r$;
    \item $m$ pairs, $(N_i.\sigma_z, \theta_z)$, of maximum influential score upper bounds and influence thresholds (for $N_i.\sigma_z = \max_{\forall v_l\in N_i} \sigma_z(hop(v_l, r))$), and;
    \item a pointer, $N_i.ptr$, pointing to a child node.
\end{itemize}

\noindent{\bf Index Construction:} To construct the tree index $\mathcal{I}$, we sorted all vertices by their average of $ub\_sup_r$ and $\sigma_z$ and recursively divided sorted vertices array into partitions of the similar sizes, and then obtain index nodes on different levels of the tree index. Then, we associate each index entry in non-leaf nodes (or each vertex in leaf nodes) with its corresponding aggregates (or pre-computed data). 

\noindent{\bf Complexity Analysis:} For the tree index $\mathcal{I}$, we denote the fanout of each non-leaf node $N$ as $\gamma$. According to our definition, since the number of leaf nodes is equal to the number of vertices $|V(G)|$, the depth of the tree index $\mathcal{I}$ is $\lceil \log_{\gamma}{|V(G)|} \rceil + 1$. On the other hand, the time complexity of recursive tree index construction is $O((\gamma^{dep}-1)/(\gamma-1))$. Therefore, the time complexity of our tree index construction is given by $O((\gamma^{\lceil \log_{\gamma}{|V(G)|} \rceil + 1}-1)/(\gamma-1))$.

For the space complexity, in the tree index $\mathcal{I}$, each non-leaf node $N$ has a list of child nodes pointers with the size $\gamma$ and $r_{max}$ tuples that contain a bit vector with size $B$, a support upper bound, $m$ pairs of threshold and score upper bound. 
Since the depth of the tree index is $\lceil \log_{\gamma}{|V(G)|} \rceil + 1$, the space complexity of all non-leaf nodes is $O(r_{max} \cdot (B+2m+\gamma) \cdot ((\gamma^{\lceil \log_{\gamma}{|V(G)|} \rceil}-1)/(\gamma-1)))$. On the other hand, each leaf node contains a keyword bit vector with size $B$, an edge support upper bound, and $m$ pairs of influential score upper bounds and influence thresholds. Thus, the space complexity of all leaf nodes is $O(|V(G)|\cdot (B+2m))$. In summary, the total space cost of our tree index is $O(r_{max} \cdot (B+2m+\gamma) \cdot ((\gamma^{\lceil \log_{\gamma}{|V(G)|} \rceil}-1)/(\gamma-1)))+O(|V(G)|\cdot (B+2m))$.

\section{Online Top$L$-ICDE Processing}
\label{sec:online_topl_icde_process}
For the online Top$L$-ICDE processing phase (see Algorithm ~\ref{algo:the_solution_framework}), we utilize the constructed tree indexes to conduct the Top$L$-ICDE processing, by integrating our effective pruning strategies and returning  top-$L$ most influential seed communities. 


\subsection{Index Pruning}
\label{subsec:index_pruning}
In this subsection, we provide effective pruning heuristics on index nodes, which can filter out all candidate seed communities under index nodes. Proofs of lemmas are omitted here due to space limitations.

\noindent{\bf Keyword Pruning for Index Entries:} We utilize the aggregated keyword bit vector, $N_i.BV_r$, of an index entry $N_i$, and discard an index entry $N_i$ if none of $r$-hop subgraphs under $N_i$ contain some keyword in the query keyword set $Q$.
\begin{lemma}
    {\bf (Index-Level Keyword Pruning)} Given an index entry $N_i$ and a set, $Q$, of query keywords, entry $N_i$ can be safely pruned, if it holds that $N_i.BV_r \wedge Q.BV = \boldsymbol{0}$, where $Q.BV$ is a bit vector hashed from the query keyword set $Q$.
    \label{lemma:index_keyword_pruning}
\end{lemma}
\begin{proof}
\label{proof:index_keyword_pruning_lemma}
The $N_i.BV_r$ is an aggregated keyword bit vector computed by $\bigvee_{\forall v_l\in N_i} v_l.BV_r$. According to the computation of keyword bit vector, the $f(w)$-th bit position valued $0$ means that any vertex $v_l\in N_i$ holds that $w \notin v_l.W$, where $w$ is a keyword and the $f(\cdot)$ is a hash function. Therefore, if $N_i.BV_r \wedge Q.BV = \boldsymbol{0}$ holds, each keyword in $Q$ is not included by any vertex belonging to $N_i$. $N_i$ can be safely pruned, which completes the proof. \qquad $\square$
\end{proof}
\nop{
\begin{proof}
Please refer to Appendix \ref{proof:index_keyword_pruning_lemma} for the detailed proof.
\end{proof}
}

\noindent{\bf Support Pruning for Index Entries:} We next use the support parameter $k$ in the $k$-truss constraint (as given in Definition~\ref{def:seed_community}) to prune an index entry $N_i$ with low edge supports.
\begin{lemma}
    {\bf (Index-Level Support Pruning)} Given an index entry $N_i$ and a support parameter $k$, entry $N_i$ can be safely pruned, if it holds that $N_i.ub\_sup_r < k$, where $N_i.ub\_sup_r$ is the maximum edge support upper bound in all $r$-hop subgraphs under $N_i$.
    \label{lemma:index_support_pruning}
\end{lemma}
\begin{proof}
\label{proof:index_support_pruning_lemma}
Since $N_i.ub\_sup_r$ is the maximum edge support upper bound in $\forall v_l \in N_i$, it holds that $N_i.ub\_sup_r \geq v_l.ub\_sup_r$ and $v_l.ub\_sup_r$ is greater than the support of any subgraph included $v_l$ satisfied truss structure. Therefore, if $N_i.ub\_sup_r < k$ holds, the support of any subgraph included $v_l$ is less than $k$ for each $v_l \in N_i$ so that $N_i$ can be safely pruned, which completes the proof. \qquad $\square$
\end{proof}
\nop{
\begin{proof}
Please refer to Appendix \ref{proof:index_support_pruning_lemma} for the detailed proof.
\end{proof}
}

\noindent{\bf Influential Score Pruning for Index Entries:} Since the Top$L$-ICDE problem finds $L$ seed communities with the highest influential scores, we can employ the following pruning lemma to rule out an index entry $N_i$ whose influential score upper bound lower than that of $L$ candidate seed communities we have seen so far.
\begin{lemma}
{\bf (Index-Level Influential Score Pruning)} Assume that we have obtained $L$ candidate seed communities with the smallest influential score $\sigma_L$. Given an index entry $N_i$ and an influence threshold $\theta \in [\theta_z, \theta_{z+1})$, entry $N_i$ can be safely pruned, if it holds that $N_i.\sigma_z \leq \sigma_L$.
\label{lemma:index_influence_pruning}
\end{lemma}
\begin{proof}
\label{proof:index_influence_pruning_lemma}
The influential score upper bound $N_i.\sigma_z$ is the maximum influential score upper bound w.r.t. threshold $\theta_z$ in $\forall v_l \in N_i$, so it holds that $N_i.\sigma_z \geq \sigma_z(hop(v_l, r)) \geq \sigma(hop(v_l, r))$, where $\sigma(hop(v_l, r))$ is the accurate influential score of $hop(v_l, r)$. The smallest influential score $\sigma_L$ holds that $\sigma_L \leq \sigma(g_i)$ for $\forall g_i \in S_L$. Therefore, if it holds that $N_i.\sigma_z \leq \sigma_L$, we obtain that $\sigma(g_i) \geq \sigma(hop(v_l, r))$ for $\forall g_i \in S_L$ and $\forall v_l \in N_i$ so that the $N_i$ can be safely pruned, which completes the proof. \qquad $\square$
\end{proof}
\nop{
\begin{proof}
Please refer to Appendix \ref{proof:index_influence_pruning_lemma} for detailed proof.
\end{proof}
}

\subsection{Top$L$-ICDE Processing Algorithm}
\label{subsec:online_topl_icde_algorithm}

Algorithm~\ref{algo:online_calculation} gives the pseudo-code to answer a Top$L$-ICDE query over a social network $G$ via the index $\mathcal{I}$. Specifically, the algorithm first initializes some data structure/variables (lines 1-4), then traverses the index $\mathcal{I}$ (lines 5-27), and finally returns actual Top$L$-ICDE query answers (line 28).

\begin{algorithm}[!ht]
\caption{{\bf Online Top$L$-ICDE Processing}}
\label{algo:online_calculation}\small
\KwIn{
    \romannumeral1) a social network $G$;
    \romannumeral2) a set, $Q$, of query keywords;
    \romannumeral3) the support, $k$, of the truss for each seed community;
    \romannumeral4) the maximum radius, $r$, of seed communities;
    \romannumeral5) the influence threshold $\theta \in [\theta_z, \theta_{z+1})$;
    \romannumeral6) an integer parameter $L$, and;
    \romannumeral7) a tree index $\mathcal{I}$ over $G$
}
\KwOut{
    a set, $S$, of top-$L$ most influential communities
}

\tcp{initialization}
hash all keywords in the query keyword set $Q$ into a query bit vector $Q.{BV}$\;

initialize a maximum heap $\mathcal{H}$ in the form of $(N, key)$\;

insert $(root(\mathcal{I}), 0)$ into heap $\mathcal{H}$\;

$S = \emptyset$; $cnt = 0$; $\sigma_L = -\infty$\;

\tcp{index traversal}

\While{$\mathcal{H}$ is not empty}{
    $(N, key)$ = de-heap$(\mathcal{H})$\;

    \If{$key \leq \sigma_L$}{terminate the loop\;}

    \eIf{$N$ is a leaf node}{
        \For{each vertex $v_i \in N$}{
            \If{$r$-hop subgraph $hop(v_i, r)$ cannot be pruned by Lemma~\ref{lemma:keyword_pruning}, \ref{lemma:support_pruning}, or \ref{lemma:influential_score_pruning}}{
                obtain seed communities $g \subseteq hop(v_i, r)$ satisfying the constraints\;
                compute influential score $\sigma(g)$ $=$ ${\sf calculate\_influence}(g, \theta)$\;
                \eIf{$cnt < L$}{
                    add $(g, \sigma(g))$ to $S$\;
                    $cnt = cnt + 1$\;
                    \If{$cnt = L$}{
                        set $\sigma_L$ to the smallest influential score in $S$\;
                    }   
                }{
                    \If{$\sigma(g) > \sigma_L$}{
                        add $(g, \sigma(g))$ to $S$\;
                        remove a candidate seed community with the lowest influential score from $S$\;
                        update influence threshold $\sigma_L$\;
                    }
                    
                }
                
            }
        }
    
    }(\tcp*[h]{$N$ is a non-leaf node}){
        \For{each entry $N_i \in N$}{
            \If{$N_i$ cannot be pruned by Lemma~\ref{lemma:index_keyword_pruning}, \ref{lemma:index_support_pruning}, or \ref{lemma:index_influence_pruning}}{
                insert entry $(N_i, N_i.\sigma_z)$ into heap $\mathcal{H}$\;
            }
        }
    }

}
\Return $S$\;
\end{algorithm}

\noindent{\bf Initialization:} Given a query keyword set $Q$, we first hash all the keywords in $Q$ into a query bit vector $Q.BV$ (line 1). Then, we use a \textit{maximum heap} $\mathcal{H}$ to traverse the index, which contains heap entries in the form of $(N, key)$, where $N$ is an index node and $key$ is the key of node $N$ (defined as maximum influential score upper bound $N.\sigma_z$, mentioned in Section \ref{subsec:indexing_mechanism}). Intuitively, if a node $N$ has a higher influential score upper bound, it is more likely that $N$ contains seed communities with high influential scores (ranks). We thus always use the maximum heap $\mathcal{H}$ to access nodes with higher influential scores earlier. We initialize heap $\mathcal{H}$ by inserting the index root in the form $(root(\mathcal{I}), 0)$ (lines 2-3). 
In addition, we use a result set, $S$, to store candidate seed communities (initialized with an empty set) whose entries are in the form of $(g, \sigma(g))$, a variable $cnt$ (initially set to $0$) to record the size of set $S$, and an influence threshold $\sigma_L$ (w.r.t. $S$, initialized to $-\infty$) (line 4).

\noindent{\bf Index Traversal:} Next, we employ the maximum heap $\mathcal{H}$ to traverse the index $\mathcal{I}$ (lines 5-27). Each time we pop out an entry $(N, key)$ with node $N$ and the maximum key, $key$, in the heap (lines 5-6). If $key$ is not greater than the smallest influential score, $\sigma_L$, among $L$ seed communities in $S$, then all entries in heap $\mathcal{H}$ have influential score upper bounds not greater than $\sigma_L$. Therefore, we can safely prune the remaining (unvisited) entries in the heap and terminate the index traversal (lines 7-8).
When we encounter a leaf node $N$, we consider $r$-hop subgraphs $hop(v_i, r)$ for all vertices $v_i$ under node $N$ (lines 9-10). Then, we apply the community-level pruning strategies, \textit{keyword pruning}, \textit{support pruning}, and \textit{influential score pruning}. If an $r$-hop subgraph $hop(v_i, r)$ cannot be pruned, we obtain seed communities, $g$, within $hop(v_i, r)$ that satisfy the constraints given in Definition~\ref{def:seed_community} and compute their accurate influential scores $\sigma(g)$ w.r.t. threshold $\theta$, by invoking the function {\sf calculate\_influence}$(g, \theta)$ (lines 11-13). Then, we will update the result set $S$ with $g$, by considering the following two cases.

\noindent\underline{\it Case 1:} If the size, $cnt$, of set $S$ is less than $L$, a new entry $(g, \sigma(g))$ will be added to $S$ (lines 14-16). 
If the set size $cnt$ reaches $L$, we will set $\sigma_L$ to the smallest influential score in $S$ (lines 17-18).

\noindent\underline{\it Case 2:} If the size, $cnt$, of set $S$ is equal to $L$ and the influential score $\sigma(g)$ is greater than $\sigma_L$, we will add the new entry $(g, \sigma(g))$ to $S$, and remove a seed community with the lowest influential score from $S$ (lines 19-22). Accordingly, threshold $\sigma_L$ will be updated with the new set $S$ (line 23).

On the other hand, when we visit a non-leaf node $N$, for each child entry $N_i \in N$, we will apply the index-level pruning strategies, including \textit{index-level keyword pruning}, \textit{index-level support pruning}, and \textit{index-level influential score pruning} (lines 24-26). If $N_i$ cannot be pruned, we insert a heap entry $(N_i, N_i.\sigma_z)$ into heap $\mathcal{H}$ for further investigation (line 27).

Finally, after the index traversal, we return actual Top$L$-ICDE query answers in $S$ (line 28).

\noindent{\bf Discussions on the Influential Score Calculation Function {\sf calculate\_influence}$(g, \theta)$:} To calculate the influential score (via Eqs.~(\ref{eq:community_to_user_influence}) and (\ref{eq:influence_score})), we need to obtain the influenced community $g^\mathit{Inf}$ from seed community $g$, whose process is similar to the single-source shortest path algorithm. We will first compute $1$-hop neighbors, $v_k$, of boundary vertices in seed community $g$, and include $v_k$ (satisfying $cpp(g, v_k) \geq \theta$) in the influenced community $g^\mathit{Inf}$. Then, each time we expand one hop from the current influenced community $g^\mathit{Inf}$, by adding to $g^\mathit{Inf}$ new vertices $v_{new}$ if it holds that $cpp(g, v_{new}) \geq \theta$, where $cpp(g, v_{new}) = \max_{\forall u\in V(g^\mathit{Inf})}\{upp(u,v_{new})\}  = \max_{\forall u\in V(g^\mathit{Inf})}\{cpp(g,u) \cdot p_{u, v_{new}}\}$.


\noindent{\bf Complexity Analysis:} 
Let $PP^{(j)}$ be the pruning power (i.e., the percentage of node entries that can be pruned) on the $j$-th level of index $\mathcal{I}$, where $0 \leq j \leq h$ (here $h$ is the height of the tree $\mathcal{I})$. Denote $f$ as the average fanout of index nodes in $\mathcal{I}$. For the index traversal, the number of nodes that need to be accessed is given by $\sum_{j=1}^{h}f^{h-j+1}\cdot(1 - {PP}^{(j)})$.
Let $\overline{n_r}$ be the average number of seed communities $g$ obtained from $hop(v_l,r)$. The BFS part of influential score computation takes $O(|V(g)| + |E(g)|)$, and the traversal (similar to the Dijkstra algorithm) implemented by the heap takes $O((|E(g^\mathit{Inf})|+|V(g^\mathit{Inf})|)log|V(g^\mathit{Inf})|)$. Moreover, the cost of the update operation on $S$ and $\sigma_L$ is a constant. Thus, it takes $O(f^{h+1}\cdot(1-{PP}^{(0)}) \cdot \overline{n_r} \cdot (|E(g^\mathit{Inf}|+|V(g^\mathit{Inf})|)log|V(g^\mathit{Inf})|)$  to refine candidate seed communities, where ${PP}^{(0)}$ is the pruning power over $r$-hop subgraphs in leaf nodes.
Therefore, the total time complexity of Algorithm~\ref{algo:online_calculation} is given by $O(\sum_{j=1}^{h}f^{h-j+1}\cdot(1 - {PP}^{(j)})+f^{h+1}\cdot(1-{PP}^{(0)}) \cdot \overline{n_r} \cdot (|E(g^\mathit{Inf}|+|V(g^\mathit{Inf})|)log|V(g^\mathit{Inf})|)$.

\section{Online Diversified Top$L$-ICDE Processing}
\label{sec:online_dtopl_icde_process}

In this section, we discuss how to efficiently tackle the Top$L$-ICDE variant, that is, the DTop$L$-ICDE problem.

\subsection{NP-Hardness of the DTop$L$-ICDE Problem}

First, we prove the NP-hardness of our DTop$L$-ICDE problem (given in Definition \ref{def:dtopl_icde_problem}) in the following lemma.
\begin{lemma}
\label{lemma:dtop_icde_np_hard}
The DTop$L$-ICDE problem is NP-hard.
\end{lemma}
\begin{proof}
\label{proof:dtop_icde_np_hard}
Consider an NP-hard Maximum Coverage problem \cite{feige1996ThresholdLnApproximating}. Given an integer $k$ and a collection of sets $\mathcal{U}=\{U_1, U_2, \ldots, U_m\}$, the goal of the Maximum Coverage problem is to select a subset $\mathcal{U}'$ ($\subseteq \mathcal{U}$) with size $k$, such that the number of elements covered by $\mathcal{U}'$, $|\bigcup_{U_i \in \mathcal{U}'}{U_i}|$, is maximized. 

A special case of our DTop$L$-ICDE problem is as follows. Given all the influenced communities, $g^{Inf}$, of $|V(G)|$ candidate seed communities $g$ (centered at vertices in $V(G)$), assume that the influential scores $\sigma(g)$ are given by the number of vertices in $g^{Inf}$ (i.e., $cpp(g, v) = 1$ for vertices $v$ within $g^{Inf}$). The DTop$L$-ICDE problem finds a subset, $S$, of $L$ influenced communities with the highest influential scores.

When $\mathcal{U}$ is a set of all the $|V(G)|$ influenced communities $g^{Inf}$ (i.e., $m= |V(G)|$), $k=L$, and $\mathcal{U}' = S$, we can reduce the \textit{Maximum Coverage} problem \cite{feige1996ThresholdLnApproximating}, which is NP-hard, to the special case of our DTop$L$-ICDE problem (as mentioned above). Hence, our DTop$L$-ICDE problem is also NP-hard, which completes the proof. \qquad $\square$

\end{proof}

\subsection{The Greedy Algorithm for DTop$L$-ICDE Processing}
\label{subsec:online_dtopl_icde_algorithm}

Due to its NP-hardness (as given by Lemma \ref{lemma:dtop_icde_np_hard}), our DTop$L$-ICDE problem is not tractable. Therefore, alternatively, we will propose a greedy algorithm to process the DTop$L$-ICDE query with an approximation bound.

\noindent {\bf A Framework for the DTop$L$-ICDE Greedy Algorithm.} Our greedy algorithm has two steps. First, we invoke \textit{online Top$L$-ICDE processing algorithm} (Algorithm~\ref{algo:online_calculation}) to obtain a set, $T$, of top-$(nL)$ candidate communities with the highest influence scores, where $n$ ($>1$) is a user-specified parameter. Intuitively, communities with high influences are more likely to contribute to the DTop$L$-ICDE community set $S$ with high diversity scores. 

Next, we will identify $L$ out of these $(nL)$ candidate communities in $T$ with high diversity score (forming a set $S$ of size $L$). To achieve this, we give a naive method of our greedy algorithm without any pruning, denoted as {\sf Greedy\_WoP}, as follows. Given $(nL)$ candidate communities with the highest influence scores in a set $T$, we first add the candidate community $g$ in $T$ with the highest influence to $S$ (removing $g$ from $T$). Then, each time we select one candidate community $g \in T$ with the highest diversity score increment $\Delta_{g}(S)$, among all communities in $T$, and move $g$ from $T$ to $S$. This process repeats until $L$ communities are added to $S$.

\noindent {\bf Effective Pruning Strategy w.r.t. Diversity Score.} In the greedy algorithm without any pruning {\sf Greedy\_WoP}, we have to check all the $nL$ candidate communities in $T$, compute their diversity score increments $\Delta_{g_m}(S)$, and select the one with the highest diversity score increment, which is quite costly with the time complexity $O(nL^2)$. 

To reduce the search space, we will propose an effective \textit{diversity score pruning} method, which can avoid scanning all communities in $T$ in each round (i.e., those communities with low diversity score increments can be safely pruned).

Before introducing our pruning strategy, we will first give two properties of the diversity score $D(S)$ below:
\begin{itemize}
    \item \textbf{Monotonicity:} given two subgraph sets $S$ and $S'$, satisfying that $S' \subseteq S$, it holds that $D(S') \leq D(S)$, and;
    \item \textbf{Submodularity:} given two subgraph sets $S$ and $S'$ and a subgraph $g$, satisfying that  $S' \subseteq S$ and $g \notin S'$, it holds that $D(S'\cup\{g\}) - D(S') \geq D(S\cup\{g\}) - D(S)$ (i.e., $\Delta D_g(S') \geq \Delta D_g(S)$).
\end{itemize}


By utilizing the two properties above, we have the following pruning lemma:
\begin{lemma}
\label{lemma:diversity_score_pruning}
{\bf (Diversity Score Pruning)} Assume that we have a set, $T$, of candidate seed communities $g$, and a set, $S$, of the currently selected DTop$L$-ICDE answers. Given a subset $S'\subseteq S$ and a subgraph $g \in T$ with the diversity score increment $\Delta D_{g}(S)$, any subgraph $g_m \in T$ can be safely pruned, if it holds that $ub\_\Delta D_{g_m}(S) < \Delta D_{g}(S)$, where $ub\_\Delta D_{g_m}(S)$ is an upper bound of the diversity score increment $\Delta D_{g_m}(S)$ (which can equal to either $\Delta D_{g_m}(S')$ or $\sigma(g_m)$).
\end{lemma}
\begin{proof}
\label{proof:diversity_score_pruning_lemma}
In the given sequence $\mathcal{S}$, it holds $S_i \subseteq S_j$ if $i < j$. According to the submodularity of the computation of the diversity score $D(S)$, it holds that $\Delta D_{g}(S_i) \geq \Delta D_{g}(S_j)$ for a candidate $g$ if $i \leq j$. Therefore, given the latest result set $S_n$, a candidate $g_m$ and any other candidate $g$, if it holds $\Delta D_{g_m}(S_n) > \Delta D_{g}(S_i)$, it means $\Delta D_{g_m}(S_n) > \Delta D_{g}(S_i) \geq \Delta D_{g}(S_n)$, which completes the proof. \qquad $\square$
\end{proof}

\nop{
\begin{proof}
Please refer to Appendix \ref{proof:diversity_score_pruning_lemma} for detailed proof.
\end{proof}
}


\noindent {\bf The DTop$L$-ICDE Greedy Algorithm With Pruning, {\sf Greedy\_WP}.}
Algorithm~\ref{algo:online_dtop_calculation} shows the pseudo-code to handle the online DTop$L$-ICDE query over a given social network $G$. Specifically, the algorithm invokes \textit{online Top$L$-ICDE processing algorithm} (i.e., Algorithm \ref{algo:online_calculation}) to obtain a set, $T$, of $nL$ candidate communities with the highest influences (line 1), then refines the set $T$ to obtain $L$ communities with the highest diversity score (lines 2-15), and finally returns the actual DTop$L$-ICDE answers (line 16).

\begin{algorithm}[!ht]
\caption{{\bf Online DTop$L$-ICDE Processing}}
\label{algo:online_dtop_calculation}\small
\KwIn{
    \romannumeral1) a social network $G$;
    \romannumeral2) a set, $Q$, of query keywords;
    \romannumeral3) the support, $k$, of the truss for each seed community;
    \romannumeral4) the maximum radius, $r$, of seed communities;
    \romannumeral5) the influence threshold $\theta \in [\theta_z, \theta_{z+1})$;
    \romannumeral6) two integer parameters $n$ and $L$, and;
    \romannumeral7) a tree index $\mathcal{I}$ over $G$
}
\KwOut{
    a set, $S$, of diversified top-$L$ most influential communities
}

\tcp{obtain ($nL$) DTop$L$-ICDE candidates}
invoke online Top$L$-ICDE processing algorithm (Algorithm~\ref{algo:online_calculation}) to obtain a set, $T$, of top-$(nL)$ most influential seed communities\;

\tcp{refine candidates via {\sf Greedy\_WP}}
initialize a maximum heap $\mathcal{H}$ with entries in the form of $(g, key_g)$\;

\For{each candidate seed community $g \in T$}{
    set $g.round = 0$\;
    insert $(g, \sigma(g))$ into heap $\mathcal{H}$\;
}

$S = \emptyset$, $round = 0$\;
\While{$|S| < L$}{
    $(g, ub\_\Delta D_{g}(S))$ = de-heap$(\mathcal{H})$\;
    \eIf{$g.round = round$}{
        add $g$ to $S$\;
        $round = round + 1$;
    }{
        compute the increment of the diversity score $\Delta D_{g}(S) = D(S\cup\{g\}) - D(S)$\;
        $g.round = round$\;
        insert $(g, \Delta D_{g}(S))$ into heap $\mathcal{H}$\;
    }
}

\Return $S$\;
\end{algorithm}


\underline{\it Initialization:} Specifically, after obtaining $nL$ candidate communities via Algorithm \ref{algo:online_calculation} (line 1), we initialize a \textit{maximum heap}, $\mathcal{H}$,  that stores entries in the form $(g, key_g)$, where $g$ is a seed community and $key_g$ is the key of the heap entry (defined as the upper bound, $ub\_\Delta D_{g}(S)$, of the diversity score increment) (line 2). For each candidate community $g\in T$, we set its round number, $g.round$, to $0$, and the upper bound $\sigma(g)$ of the diversity score increment in this round. Then, we insert entries $(g, \sigma(g))$ into heap $\mathcal{H}$ for the refinement (lines 3-5). We also maintain an initially empty answer set $S$, and set the initial round number, $round$, to $0$ (line 6).



\underline{\it Candidate Community Refinement:} To refine candidate communities in heap $\mathcal{H}$, each time we pop out an entry $(g, ub\_\Delta D_{g}(S))$ from $\mathcal{H}$ with the maximum key (line 8), and check whether or not the key $ub\_\Delta D_{g}(S)$ is computed at this round (i.e., $g.round = round$), considering the following two cases (lines 9-15):

\underline{\it Case 1:} If $g.round = round$ holds, it indicates that $g$ is the one with the highest diversity score increment in the current round, $round$ (as proved by Lemma~\ref{lemma:diversity_score_pruning}). Thus, $g$ will be added to the DTop$L$-ICDE answer set $S$ and $round$ is increased by 1 (lines 9-11).

\underline{\it Case 2:} If $g.round \neq round$ holds (line 12), the entry key, $ub\_\Delta D_{g}(S)$, is outdated, which is equal to $\Delta D_{g}(S')$ ($S'\subseteq S$). Thus, we will re-compute the diversity score increment, $\Delta D_g(S)$, w.r.t. the current answer set $S$, update $g.round$ with $round$, and insert $(g, \Delta D_g(S))$ back into $\mathcal{H}$ (lines 13-15).


After picking $L$ candidates from $\mathcal{H}$ to $S$, the algorithm terminates the loop (line 7) and returns $S$ as DTop$L$-ICDE answers (line 16).

\noindent{\bf Complexity Analysis:}
To obtain top-$(nL)$ influential communities, the total time complexity is given in Section~\ref{subsec:online_topl_icde_algorithm}. For the refinement, we create a heap for $(nL)$ candidates. We need to compute the diversity score increments for $nL^2$ times in the worst case and $L$ times for the best case. Let $DPP^{(k)}$ be the diversity pruning power (i.e., the percentage of candidates that can be pruned) in the $k$-th iteration, where $0 \leq k \leq L$. Therefore, the computation of the diversity score increment is executed for $\sum_{k=1}^{L}(nL-k+1) \cdot (1 - DPP^{(k)})$ times. Since the computation cost of the diversity score increment is $O(|\bigcup_{g \in T}V(g^\mathit{Inf})|)$, where $T$ is the set of top-$(nL)$ influential communities, the total time complexity of the refinement is given by $O\left(\sum_{k=1}^{L}(nL-k+1) \cdot (1 - DPP^{(k)}) \cdot |\bigcup_{g \in T}V(g^\mathit{Inf})|\right)$.


\noindent{\bf Approximation Ratio Analysis:} From \cite{kempe2003MaximizingSpreadInfluence}, for a function following \textit{monotonicity} and \textit{submodularity} properties, the approximate greedy algorithm has a $(1-1/e)$ approximation guarantee. Moreover, we can prove that our greedy algorithm has a $\epsilon\cdot(1-1/e)$ approximation guarantee, where $0 < \epsilon \leq 1$.

\begin{lemma}
\label{lemma:dtop_icde_approximation}
The online DTop$L$-ICDE processing algorithm can process the DTop$L$-ICDE query approximately within better than a factor of $\epsilon\cdot(1-1/e)$, where $0 < \epsilon \leq 1$. 
\end{lemma}
\begin{proof}
\label{proof:dtop_icde_approximation_lemma}
With the naive greedy algorithm, the candidate set, $\hat{S}$, contains all of the seed communities, while the online DTop$L$-ICDE processing algorithm only selects the candidate set, $S^\prime$, with the top-$nL$ influential score. Since that $S^\prime \subseteq \hat{S}$, it is not difficult to prove that the diversity score of the result set, $D(S)$, must not be smaller than $({|S^\prime|}/{|\hat{S}|})\cdot(1-1/e)$ multiplying the diversity score of the optimal result set. Denoting the ${|S^\prime|}/{|\hat{S}|}$ as $\epsilon$, where $0 < \epsilon \leq 1$, we can conclude an approximation ratio of $\epsilon\cdot(1-1/e)$ for online DTop$L$-ICDE processing algorithm. \qquad $\square$
\end{proof}
\nop{
\begin{proof}
Please refer to Appendix \ref{proof:dtop_icde_approximation_lemma} for detailed proof.
\end{proof}
}



\section{Experimental Evaluation}
\label{sec:experiments}

\subsection{Experimental Settings}

We tested the performance of our Top$L$-ICDE processing approach (i.e., Algorithm~\ref{algo:online_calculation}) on both real and synthetic graphs. Our code is uploaded on Github at the URL (\url{https://github.com/Watremons/influential-community-detection})

\noindent {\bf Real-World Graph Data Sets:}
We used two real-world graphs, \textit{DBLP} and \textit{Amazon}, similar to previous works \cite{luo2023EfficientInfluentialCommunity, zhou2023InfluentialCommunitySearch, chen2015OnlineTopicawareInfluence}, whose statistics are depicted in Table~\ref{tab:real_data_set}. \textit{DBLP} is a bibliographical network, in which two authors are connected if they co-authored at least one paper, whereas \textit{Amazon} is an \textit{Also Bought} network where two products are connected if they are co-purchased by customers.






\noindent {\bf Synthetic Graph Data Sets:}
For synthetic social networks, we generate \textit{Newman–Watts–Strogatz small-world} graphs $G$ \cite{newman1999RenormalizationGroupAnalysis}. Specifically, we first produce a ring with size $|V(G)|$, and then connect each vertex with its $m$ nearest neighbors in the ring. Next, for each resulting edge $e_{u,v}$, with probability $\mu$, we add a new edge $e_{u,w}$ between $u$ and a random vertex $w$. Here, we set $m=6$ and $\mu = 0.167$. For each vertex, we also randomly produce a keyword set $v_i.W$ from the keyword domain $\Sigma$, following Uniform, Gaussian, or Zipf distribution, and obtain three synthetic graphs, denoted as \textit{Uni}, \textit{Gau}, and \textit{Zipf}, respectively. For each edge $e_{u,v}$ in graph $G$, we randomly generate a value within the interval $[0.5,0.6)$ as the edge weight $p_{u,v}$. The propagation probabilities can be computed based on the MIA model (Section \ref{subsec:information_propagation_model}). 


In our experiments, we randomly select $|Q|$ keywords from the keyword domain $\Sigma$ and form a query keyword set $Q$.


\noindent \textbf{Competitors:}
To our best knowledge, no prior works studied the Top$L$-ICDE problem and its variant DTop$L$-ICDE problem by considering highly connected $k$-truss communities with user-specified keywords and high (collective) influences on other users. Thus, for Top$L$-ICDE, we use a baseline method, named \textit{ATindex}, which applies the state-of-the-art $(k,d)$-truss community search algorithm \cite{huang2017attribute}. Specifically, \textit{ATindex} offline pre-computes and indexes the trussness on vertices and edges. Then, it online filters out vertices with trussness less than $k$ via the index, extracts $r$-hop subgraphs (satisfying the keyword constraints w.r.t. $Q$) centered at the remaining vertices, and obtains maximal $k$-truss within $r$-hop subgraphs. After that, \textit{ATindex} computes influential scores of these $k$-truss subgraphs and returns $L$ communities with the highest influential scores.

For DTop$L$-ICDE, we compare our approach (using {\sf Greedy\_WP}) with {\sf Greedy\_WoP} and \textit{Optimal} 
methods. {\sf Greedy\_WoP} is the greedy algorithm without pruning mentioned in Section~\ref{subsec:top_L_most_influential_community_detection}, whereas \textit{Optimal} computes the diversity score for each possible combination of seed communities and selects the one with the maximum diversity score.

\begin{table}[t]
\begin{center}
\caption{Statistics of real-world graph data sets \textit{DBLP} and \textit{Amazon}.}
\label{tab:real_data_set}
\footnotesize
\begin{tabular}{|l||l||l|}
\hline
\textbf{Social Networks} & $|V(G)|$ & $|E(G)|$\\
\hline\hline
    \textit{DBLP} & 317,080 & 1,049,866\\\hline
    \textit{Amazon} & 334,863 & 925,872\\\hline
\end{tabular}
\end{center}\vspace{-2ex}
\end{table}

\begin{table}[t]
\begin{center}
\caption{Parameter settings.}
\label{tab:parameters}\vspace{-3ex}
\footnotesize
\begin{tabular}{|l||p{30ex}|}
\hline
\textbf{Parameters}&\textbf{Values} \\
\hline\hline
    influence threshold $\theta$ & 0.1, \textbf{0.2}, 0.3\\\hline
    size, $|Q|$, of query keyword set $Q$  & 2, 3, \textbf{5}, 8, 10 \\\hline
    support, $k$, of truss structure & 3, \textbf{4}, 5\\\hline    
    radius $r$ & 1, \textbf{2}, 3\\\hline
    size, $L$, of query result set & 2, 3, \textbf{5}, 8, 10\\\hline
    size, $|v_i.W|$, of keywords per vertex & 1, 2, {\bf 3}, 4, 5 \\\hline
    keyword domain size $|\Sigma|$ & 10, \textbf{20}, 50, 80\\\hline
    the size, $|V(G)|$, of data graph $G$ & 10K, 25K, {\bf 50K}, 100K, 250K, 500K, 1M\\\hline
    parameter, $n$, for DTop$L$-ICDE & 2, 3, \textbf{5}, 8, 10\\\hline
\end{tabular}
\end{center}\vspace{-2ex}
\end{table}

\begin{figure}[t!]
    \centering
    \includegraphics[height=4cm]{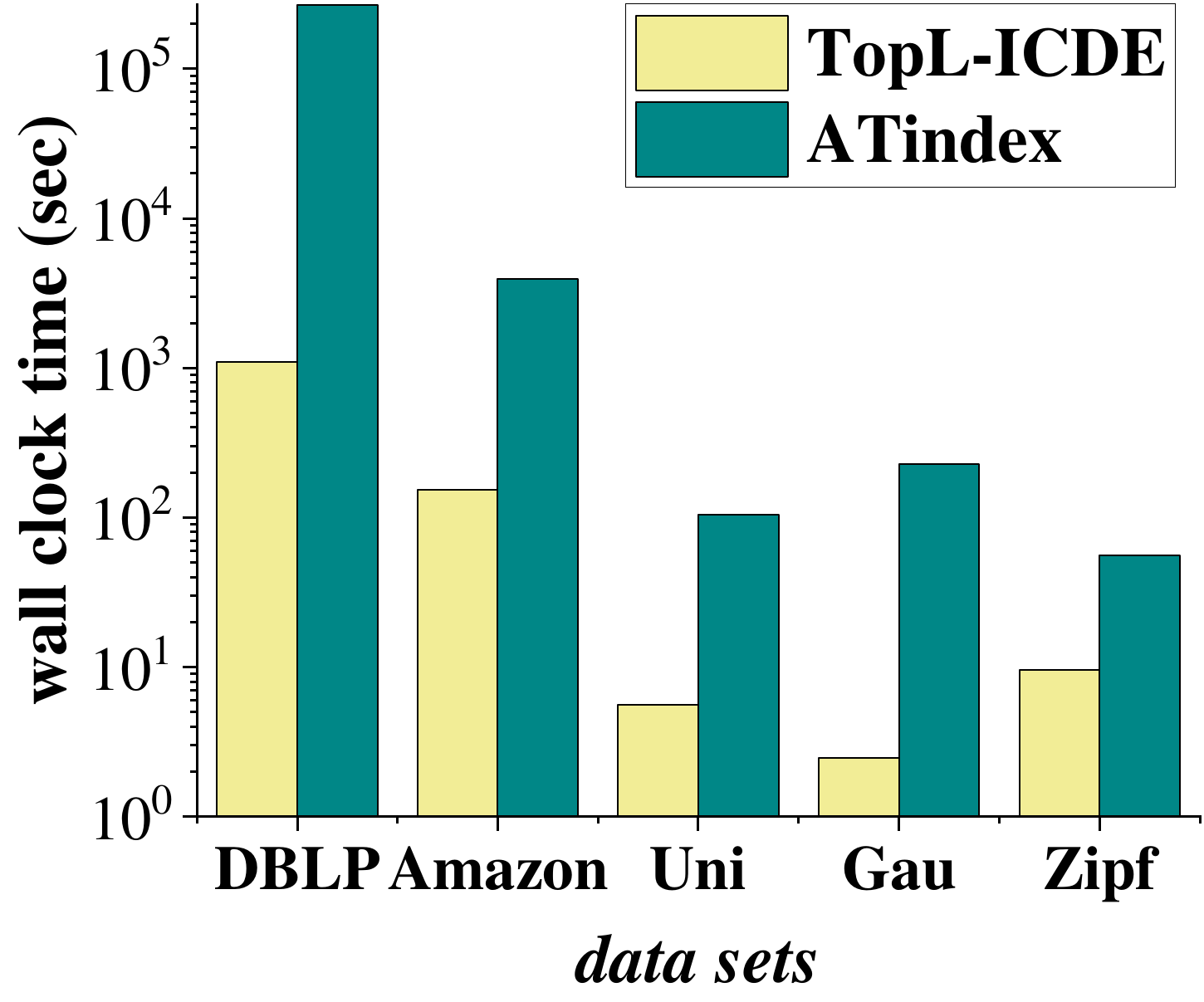}
    \caption{The Top$L$-ICDE performance on real/synthetic graph data.}
    \label{fig:performance_total}
\end{figure}

\noindent \textbf{Measures:}
To evaluate the efficiency of our Top$L$-ICDE approach, we report the \textit{wall clock time}, which is the time cost to online retrieve Top$L$-ICDE answers via the index (Algorithm~\ref{algo:online_calculation}). For DTop$L$-ICDE, we report the \textit{wall clock time} and \textit{accuracy} (defined as the ratio of the diversity score of our method to that of the optimal method).

\noindent \textbf{Parameters Settings:}
Table~\ref{tab:parameters} depicts the parameter settings, where default values are in bold. Each time we vary the values of one parameter, while other parameters are set to their default values. We ran all the experiments on the machines with Intel(R) Core(TM) i9-10900K 3.70GHz CPU, Ubuntu 20.04 OS, and 32 GB memory. All algorithms were implemented in Python and executed with Python 3.9 interpreter.

\noindent \textbf{Research Questions:} We conduct extensive experiments to evaluate our Top$L$-ICDE and DTop$L$-ICDE approaches and answer the following four research questions (RQs):

\underline{\it RQ1 (Efficiency):} Can our proposed approaches efficiently process Top$L$-ICDE and DTop$L$-ICDE queries? 

\underline{\it RQ2 (Effectiveness):} Can our proposed pruning strategies effectively filter out candidate communities during Top$L$-ICDE query processing? 

\underline{\it RQ3 (Meaningfulness):} Are the resulting Top$L$-ICDE communities useful for real-world applications?

\underline{\it RQ4 (Accuracy):} Can our proposed approach achieve high accuracy of DTop$L$-ICDE query answers?

\subsection{Top$L$-ICDE Performance Evaluation}

\setlength{\textfloatsep}{0pt}
\begin{figure*}[t!]
    \centering
    \subfigure[influence threshold $\theta$]{
        \includegraphics[height=3.2cm]{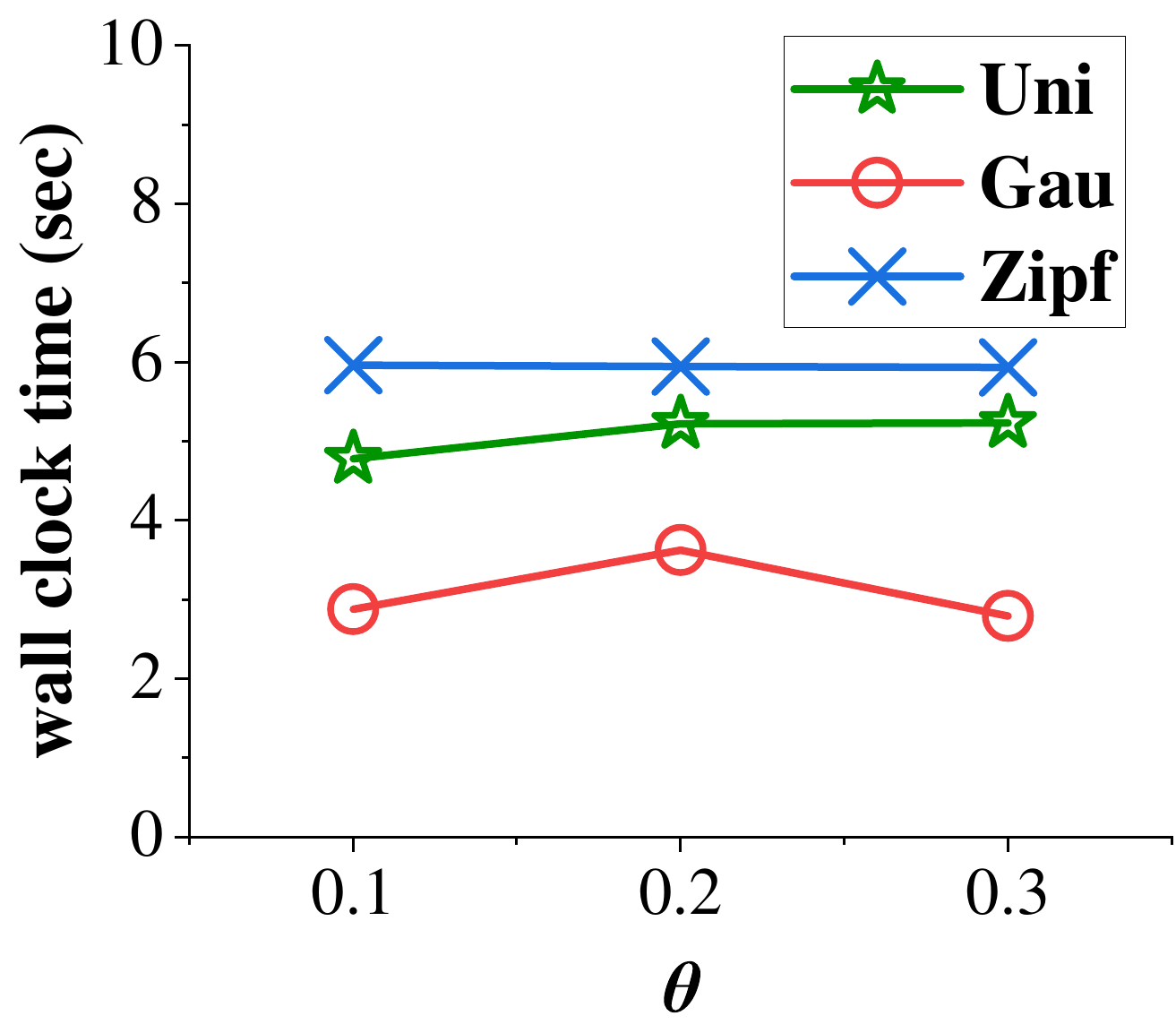}
        \label{subfig:effect_influence_threshold}
    }
    \hspace{-0.3cm}
    \subfigure[query keyword set size $|Q|$]{
        \includegraphics[height=3.2cm]{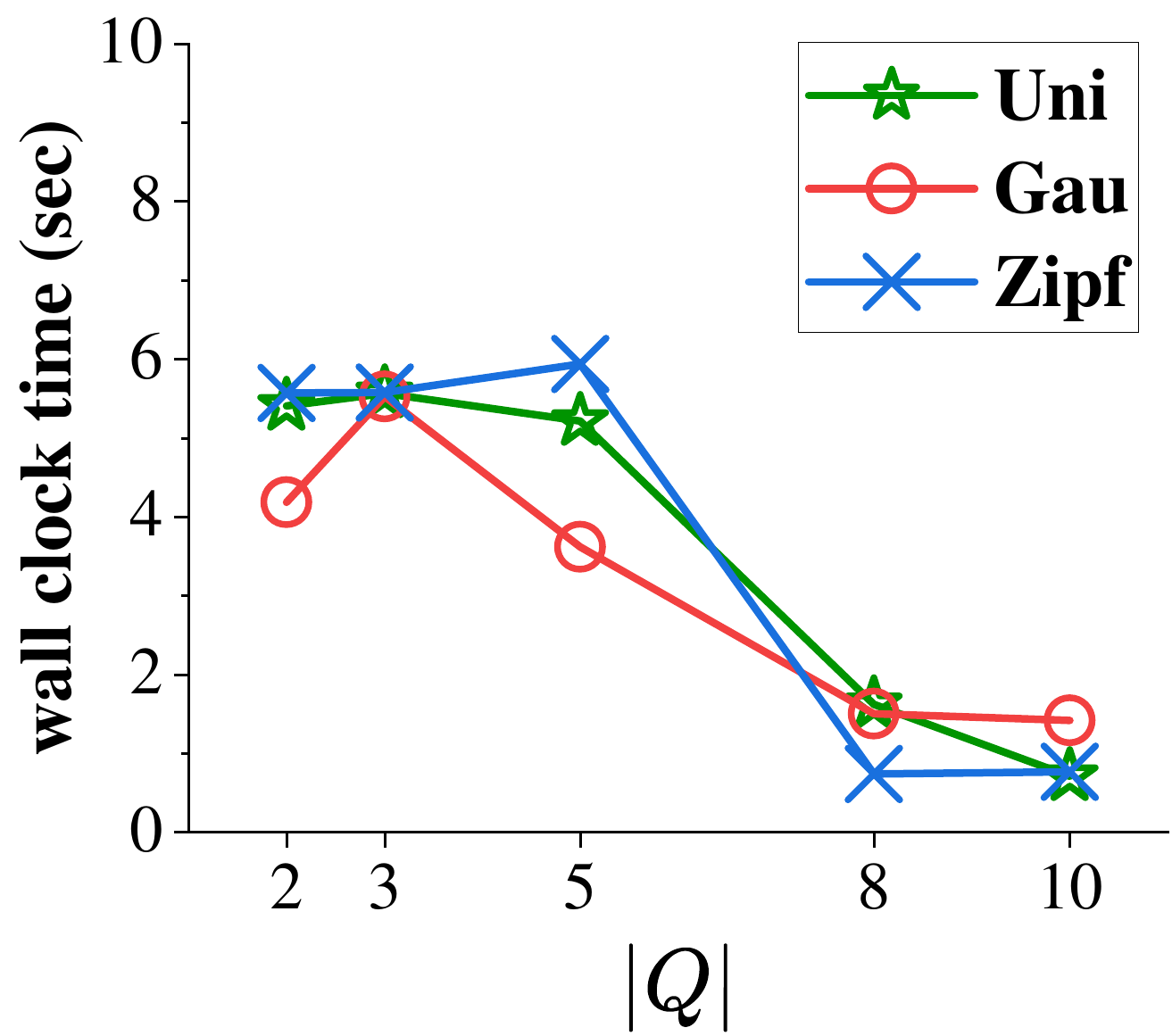}
        \label{subfig:effect_query_keywords_set_size}
    }
    \subfigure[truss support param. $k$]{
        \includegraphics[height=3.2cm]{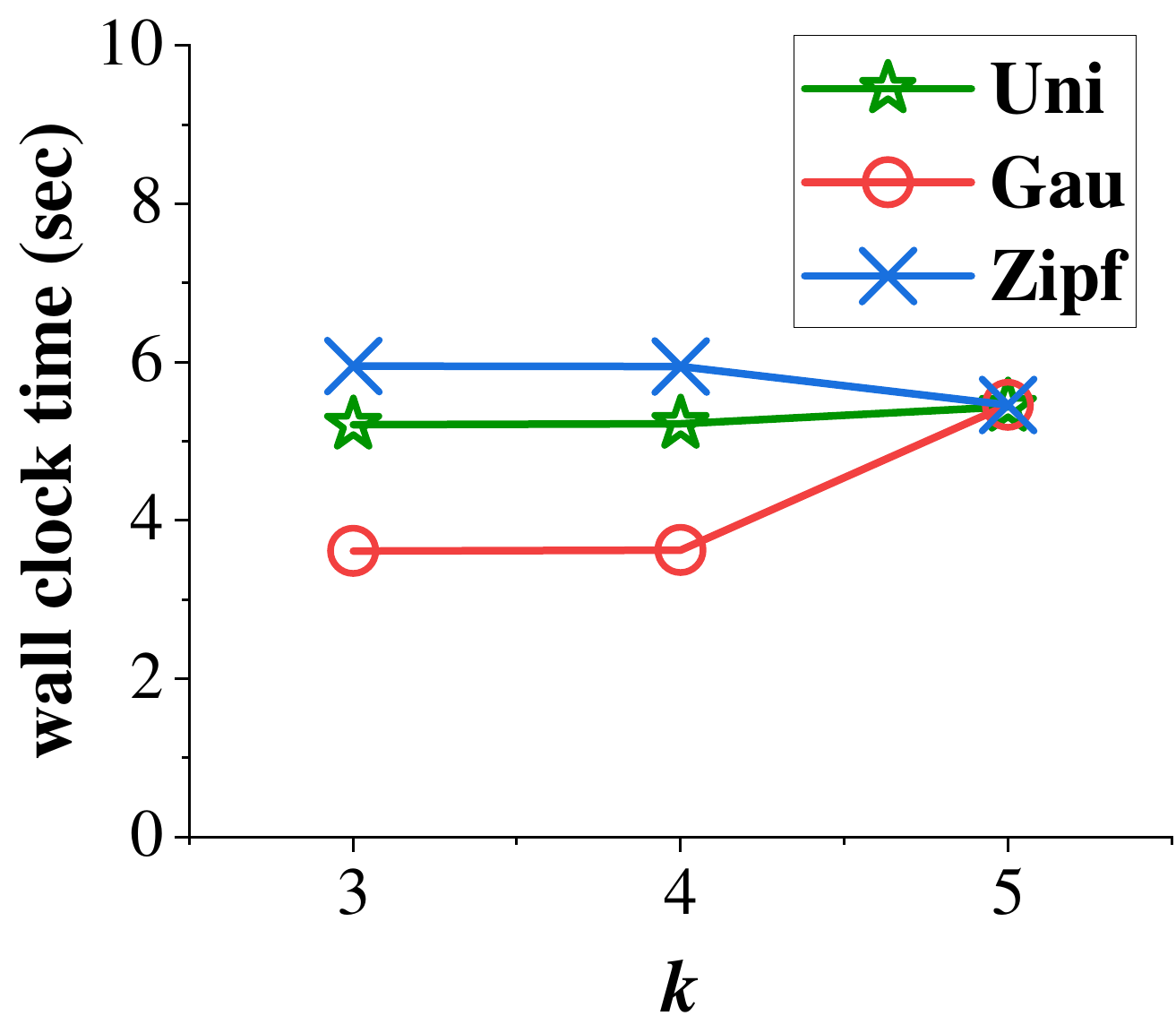}
        \label{subfig:effect_truss_support_parameter}
    }
    \hspace{-0.3cm}
    \subfigure[radius $r$]{
        \includegraphics[height=3.2cm]{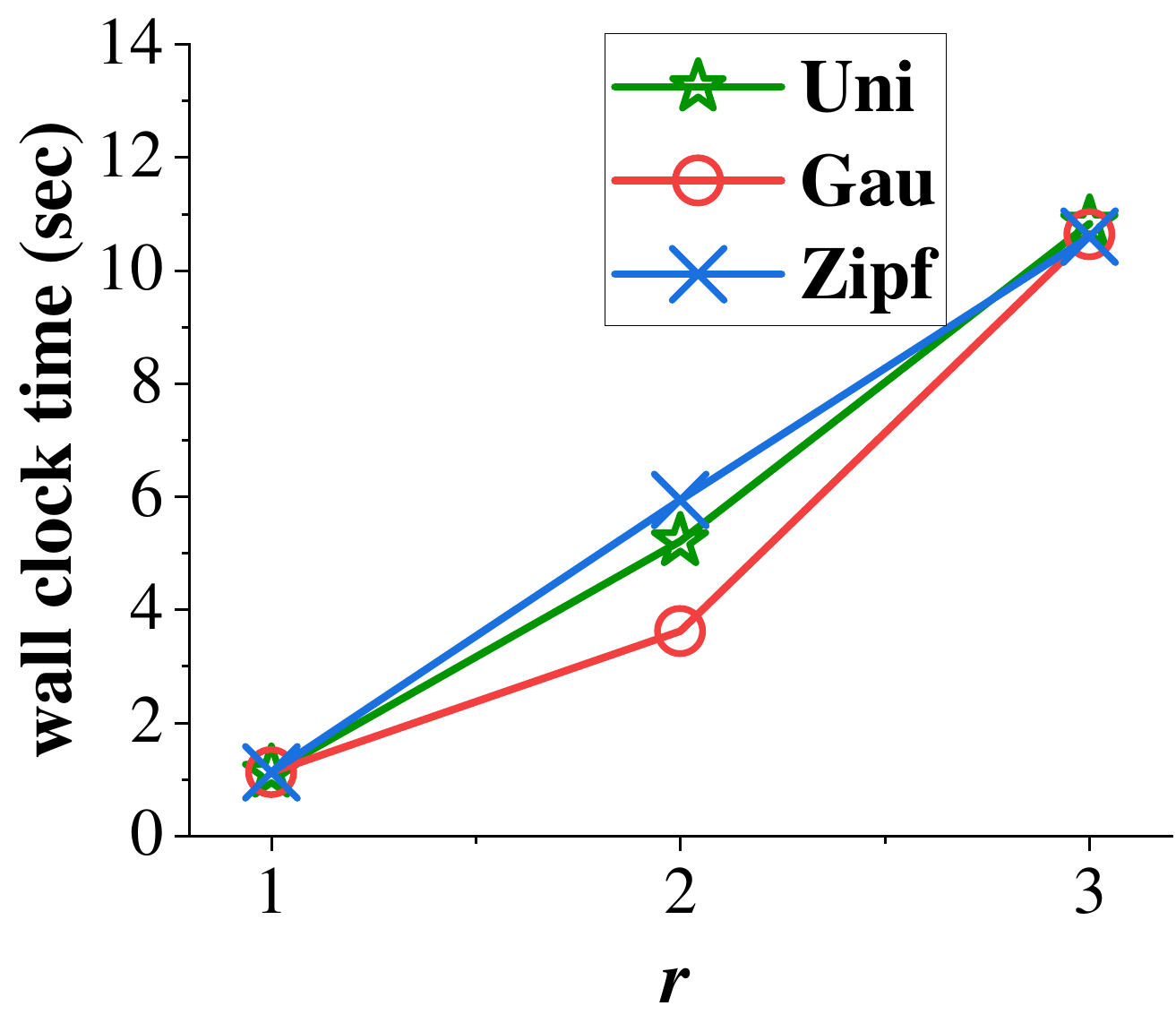}
        \label{subfig:effect_radius}
    }\\
    \subfigure[query result size $L$]{
        \includegraphics[height=3.2cm]{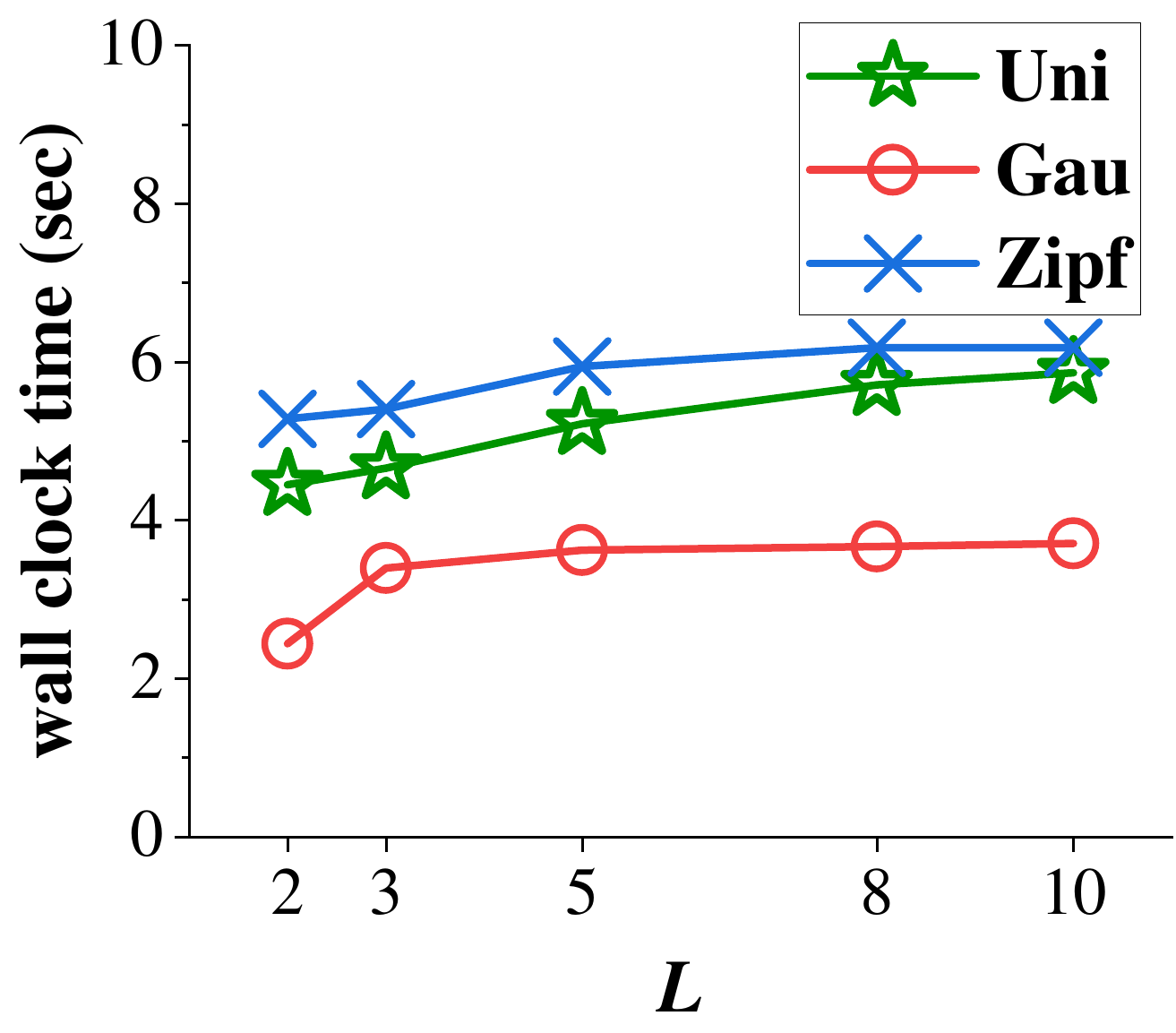}
        \label{subfig:effect_result_set_size}
    }
    \hspace{-0.3cm}
    \subfigure[\# of keywords/vertex $|v_i.W|$]{
        \includegraphics[height=3.2cm]{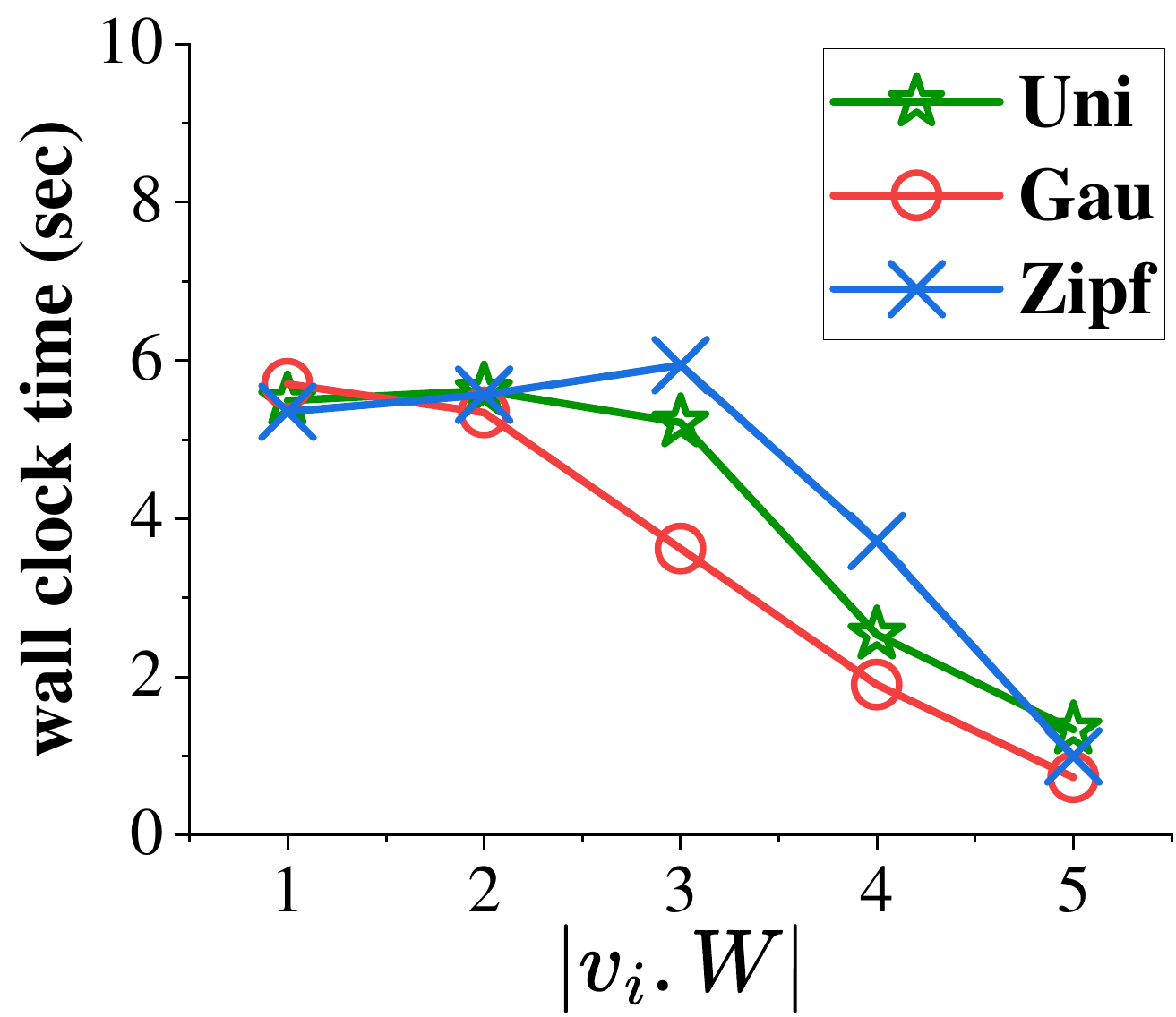}
        \label{subfig:effect_keywords_per_vertex_size}
    }
    \subfigure[keyword domain size $|\Sigma|$]{
        \includegraphics[height=3.2cm]{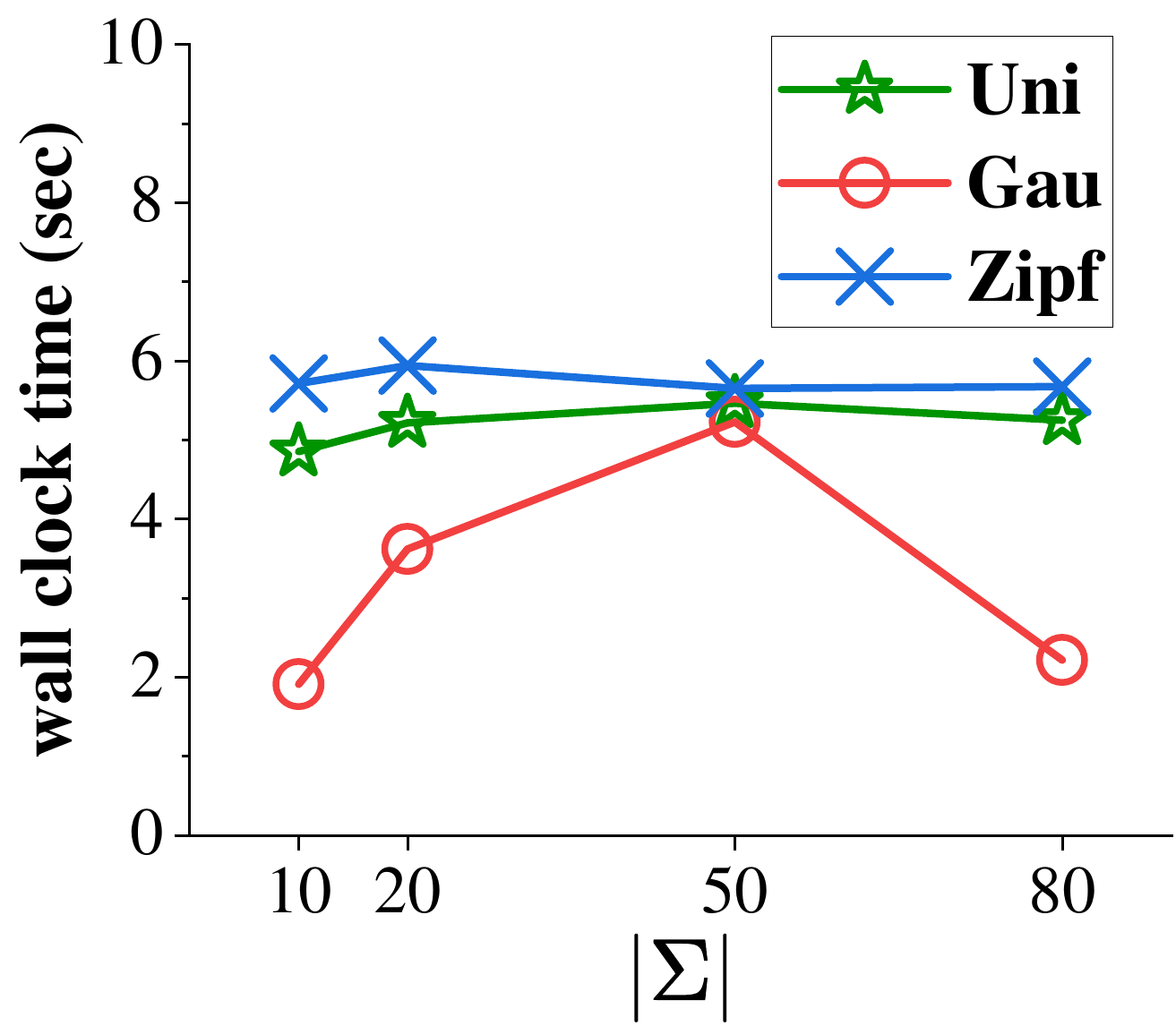}
        \label{subfig:effect_keywords_domain_size}
    }
    \hspace{-0.4cm}
    \subfigure[data graph size $|V(G)|$]{
        \includegraphics[height=3.2cm]{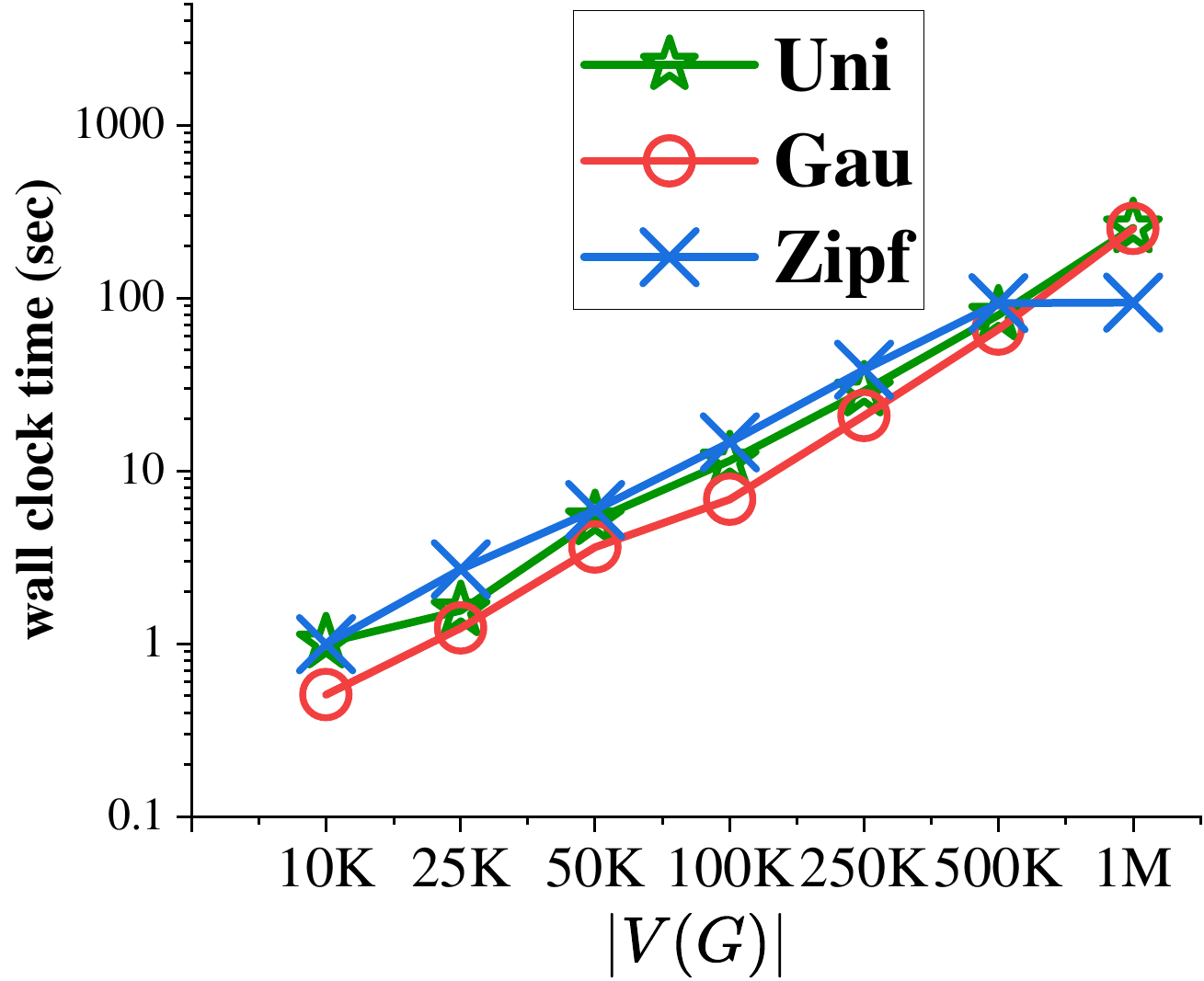}
        \label{subfig:effect_data_graph_size}
    }
    \caption{The robustness evaluation of the Top$L$-ICDE performance.}
    \label{fig:efficiency}
\end{figure*}

\begin{figure}[t!]\vspace{-2ex}
    \centering
    \subfigure[pruning power] {
        \includegraphics[height=2.8cm]{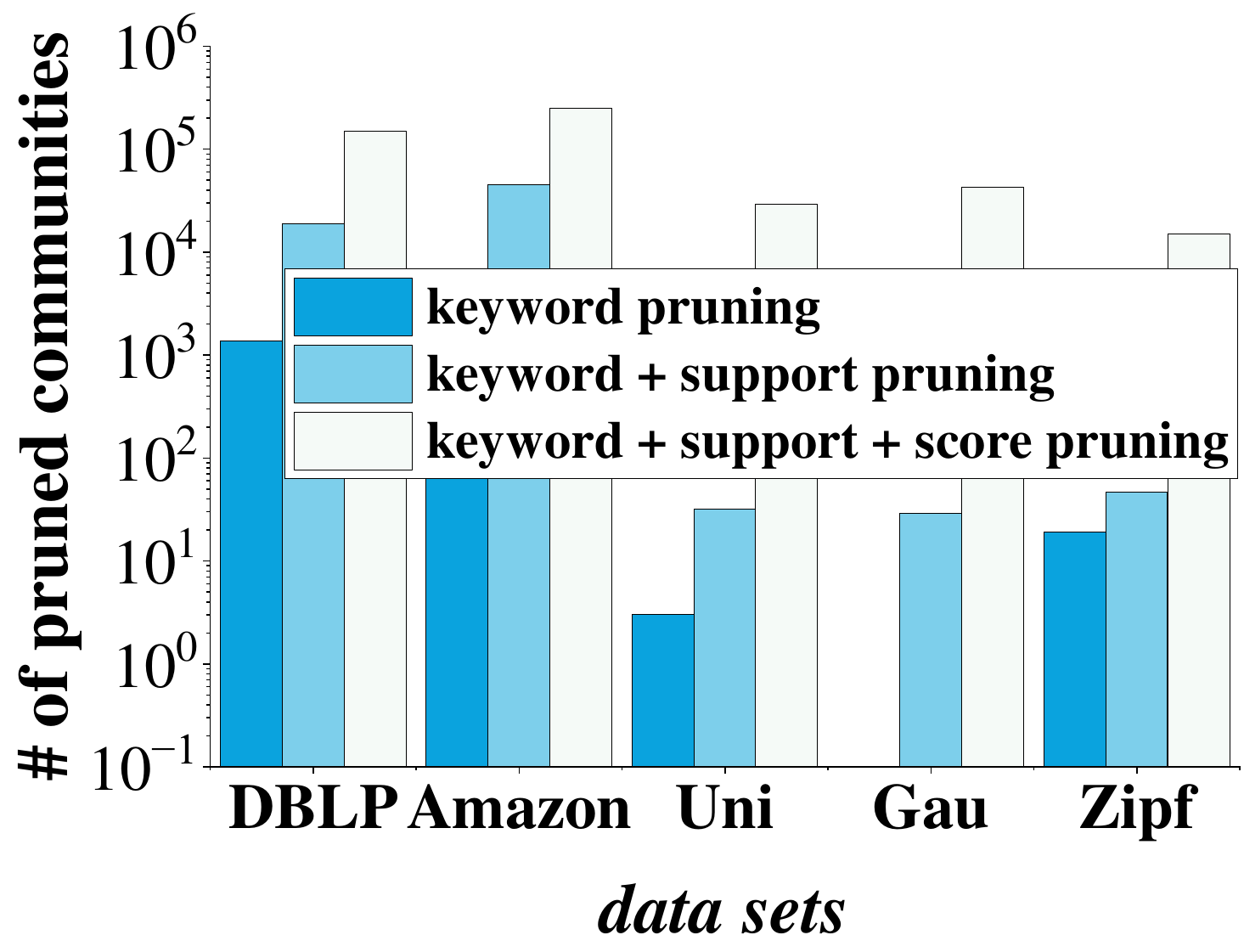}
        \label{subfig:ablation_pruning_power}
    }\hspace{-0.3cm}
    \subfigure[time cost]{
        \includegraphics[height=2.8cm]{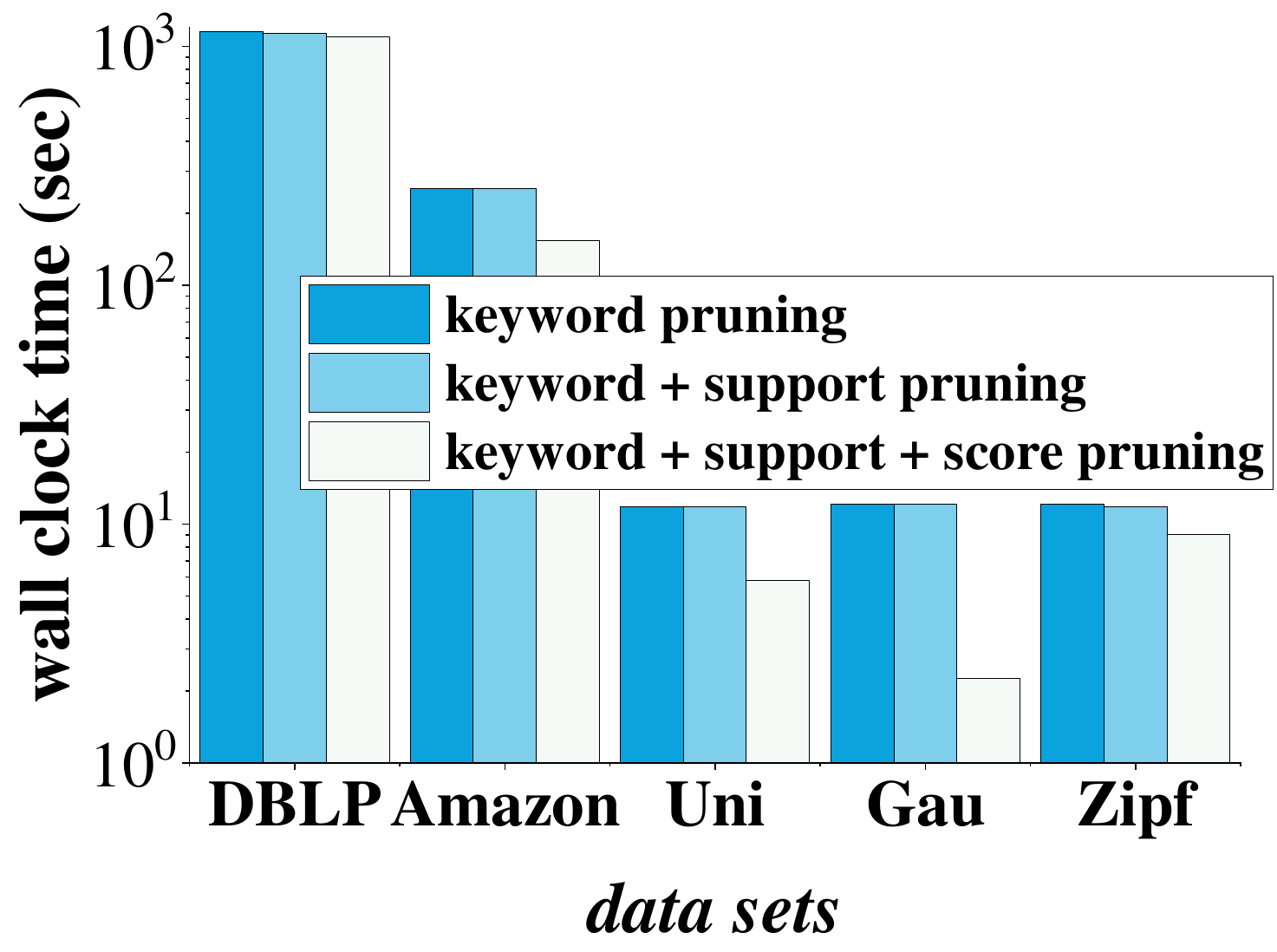}
        \label{subfig:ablation_timecost}
    }
    \vspace{-1ex}
    \caption{The ablation study of the Top$L$-ICDE performance.}
    \vspace{-4ex}
    \label{fig:ablation}
\end{figure}

\begin{figure}[t!]
\subfigcapskip=-0.2cm
    \centering
    \subfigure[Top$L$-ICDE ($\sigma(g)=344.31$,\newline 974 possibly influenced nodes)] {
        \includegraphics[height=3cm]{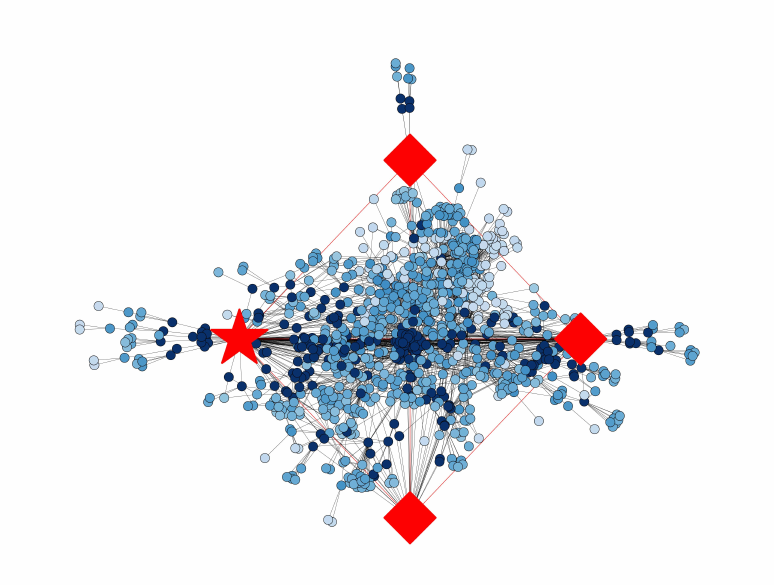}
        \label{subfig:k_truss_case_study}
    }
    \subfigure[$k$-core ($\sigma(g)=239.81$,\newline  646 possibly influenced nodes)]{
        \includegraphics[height=3cm]{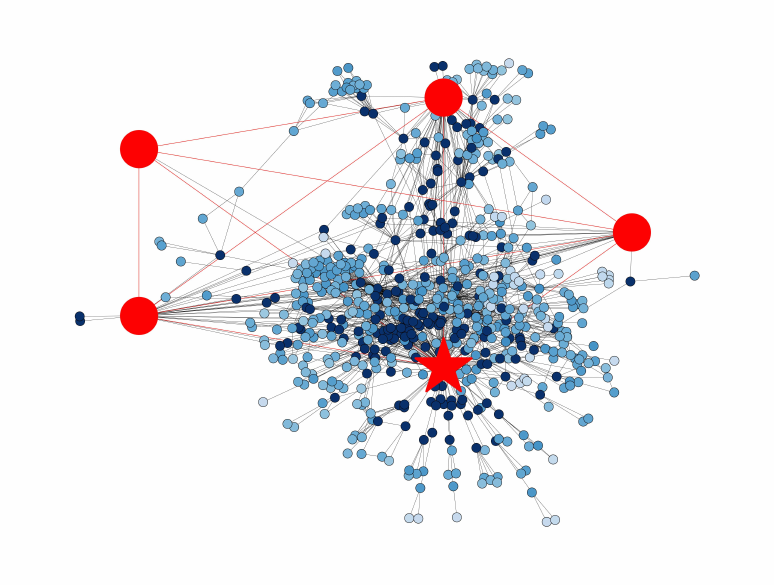}
        \label{subfig:k_core_case_study}
    }
    \vspace{-1ex}
    \caption{The influenced communities from Top$L$-ICDE vs. $k$-core ($k=4$).}
    \label{fig:case_study}
\end{figure}

\noindent \textbf{The Top$L$-ICDE Efficiency on Real/Synthetic Graphs (RQ1):}
Figure~\ref{fig:performance_total} compares the performance of our Top$L$-ICDE approach with that of \textit{ATindex} over real and synthetic graphs, in terms of the \textit{wall clock time}, where we set all parameters to their default values in Table~\ref{tab:parameters}. 
Note that, for \textit{DBLP}, since the time cost of \textit{ATindex} is extremely high, we sample $0.5\%$ center vertices from original graph data without replacement and estimate the total time as $\frac{\overline{t_s}}{0.005}=200\cdot\overline{t_s}$, where $\overline{t_s}$ is the average time per sample. The experimental results show that our Top$L$-ICDE approach outperforms \textit{ATindex} by more than one order of magnitude, which confirms the efficiency of our Top$L$-ICDE algorithm on real/synthetic graphs.






To evaluate the robustness of our Top$L$-ICDE approach, in subsequent experiments, we will vary different parameters (e.g., $\theta$, $|Q|$, $k$, $r$, $L$) on synthetic graphs.

\underline{\it Effect of Influence Threshold $\theta$:} Figure~\ref{subfig:effect_influence_threshold} shows the performance of our Top$L$-ICDE approach, by varying the influence threshold $\theta$ from $0.1$ to $0.3$, where other parameters are by default. When $\theta$ increases, the sizes of candidate seed communities become smaller, which leads to smaller influential score bound $\sigma_L$ (i.e., lower pruning power) and thus higher filtering cost. On the other hand, with larger $\theta$, the cost of refining smaller candidate seed communities becomes lower. Therefore, in Figure~\ref{subfig:effect_influence_threshold}, for larger $\theta$, the wall clock time first increases and then decreases for all the three synthetic graphs. Nevertheless, for different $\theta$ values, the time costs remain low (i.e., 2.79 $\sim$ 5.96 $sec$) over all the graphs.

\underline{\it Effect of the Size, $|Q|$, of the Query Keyword Set:}
Figure~\ref{subfig:effect_query_keywords_set_size} varies the size, $|Q|$, of query keywords set $Q$ from $2$ to $10$, and reports the \textit{wall clock time} of our Top$L$-ICDE approach, where default values are used for other parameters. With the increase of $|Q|$, more candidate seed communities are detected, which incurs higher influential score bound $\sigma_L$ and in turn higher pruning power. Thus, when $|Q|$ increases, almost all curves decrease, except for smaller wall clock time when $|Q| < 5$. This is due to empty or small ($<L$) answer sets resulting from small $|Q|$. Nonetheless, for different $|Q|$ values, the \textit{wall clock time} remains low (i.e., 0.71 $\sim$ 5.94 $sec$).


\underline{\it Effect of Truss Support Parameter $k$:}
Figure~\ref{subfig:effect_truss_support_parameter} illustrates the performance of our Top$L$-ICDE approach, where the support parameter of the truss $k = 3$, $4$, and $5$, and default values are used for other parameters. The time cost is generally not very sensitive to $k$ values, since edge supports are similar in all the three synthetic graphs. When $k = 5$, however, no candidate communities can be detected, and thus time costs on the three graphs are similar (but trends are different from cases of $k=3, 4$). For different $k$ values, wall clock times of our Top$L$-ICDE approach remain low (i.e., 3.61 $\sim$ 5.95 $sec$).

\underline{\it Effect of Radius $r$:}
Figure~\ref{subfig:effect_radius} illustrates the experimental results of our Top$L$-ICDE approach for different radii $r$ of seed communities, where $r = 1$, $2$, and $3$ and other parameters are by default. A larger radius $r$ leads to larger seed communities to filter and refine, which incurs higher time costs, as confirmed by the figure. Nonetheless, the \textit{wall clock time} remains small (i.e., around 1.12 $\sim$ 10.83 $sec$) for different $r$ values.


\underline{\it Effect of the Size, $L$, of Query Result Set:}
Figure~\ref{subfig:effect_result_set_size} presents the performance of our Top$L$-ICDE approach with different sizes, $L$, of query result set, where $L$ varies from $2$ to $10$ and default values are used for other parameters. Intuitively, the larger $L$, the more communities must be processed. Despite that, the time cost of Top$L$-ICDE remains low (i.e., 2.44 $\sim$ 6.18 $sec$) for different $L$ values.

\underline{\it Effect of the Number, $|v_i.W|$, of Keywords per Vertex:}
Figure~\ref{subfig:effect_keywords_per_vertex_size} varies the average number, $|v_i.W|$, of keywords per vertex from $1$ to $5$, where other parameters are set to their default values. More keywords in $v_i.W$ imply higher probabilities that vertex $v_i$ is included in candidate seed communities, which leads to higher processing costs. On the other hand, larger seed communities will have higher influential score bound $\sigma_L$ and in turn higher pruning power. Thus, the wall clock time is co-affected by these two factors. For larger $|v_i.W|$, the time cost first increases and then decreases but nevertheless remains low  (i.e., 0.73 $\sim$ 5.94 $sec$).


\setlength{\textfloatsep}{0pt}
\begin{figure*}[t!]
    \centering
    \subfigure[time cost vs. graphs]{
        \includegraphics[height=2.8cm]{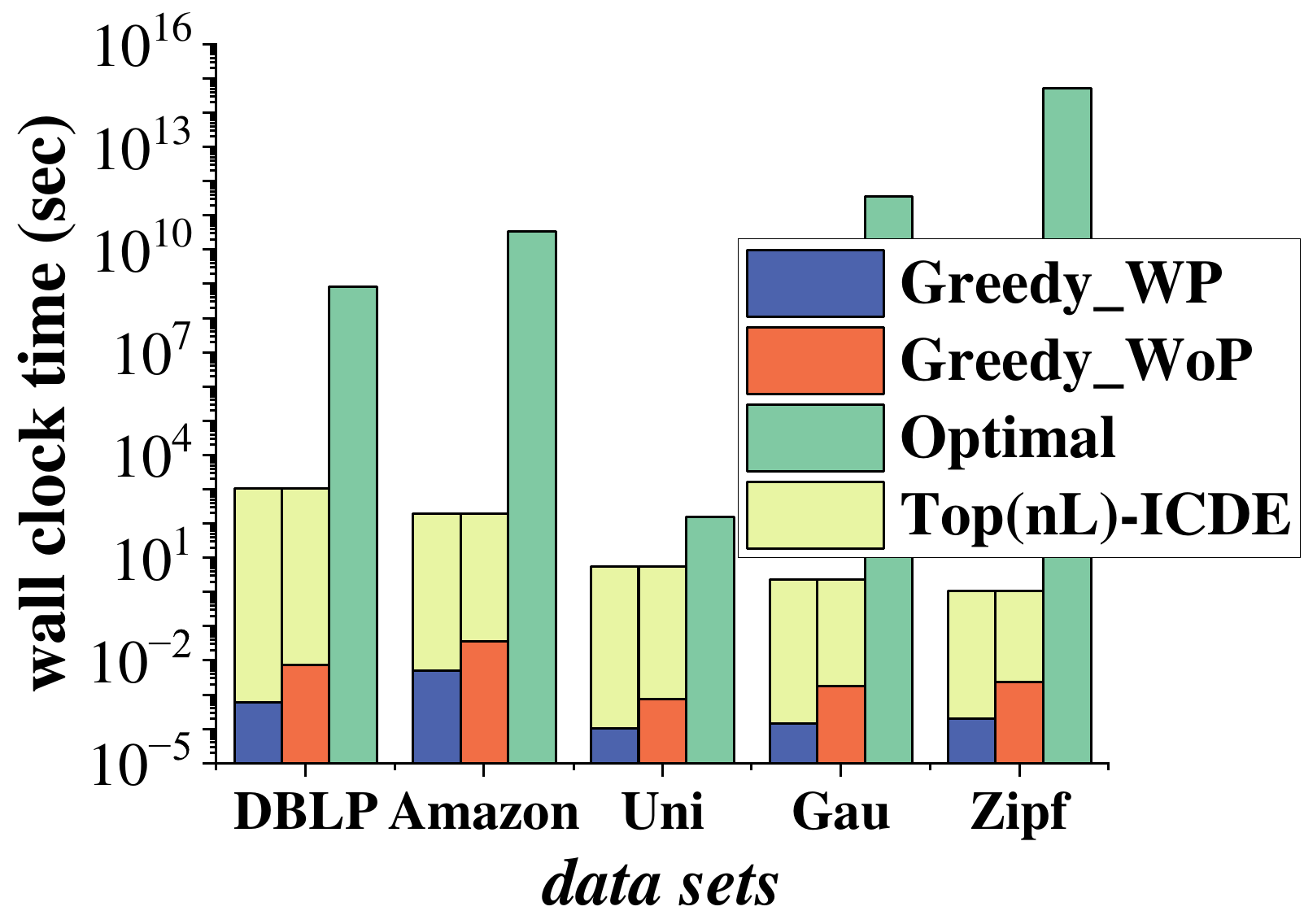}
        \label{subfig:dtopl_performance_time}
    }
    \hspace{-0.3cm} 
    \subfigure[query result size $L$]{
        \includegraphics[height=2.8cm]{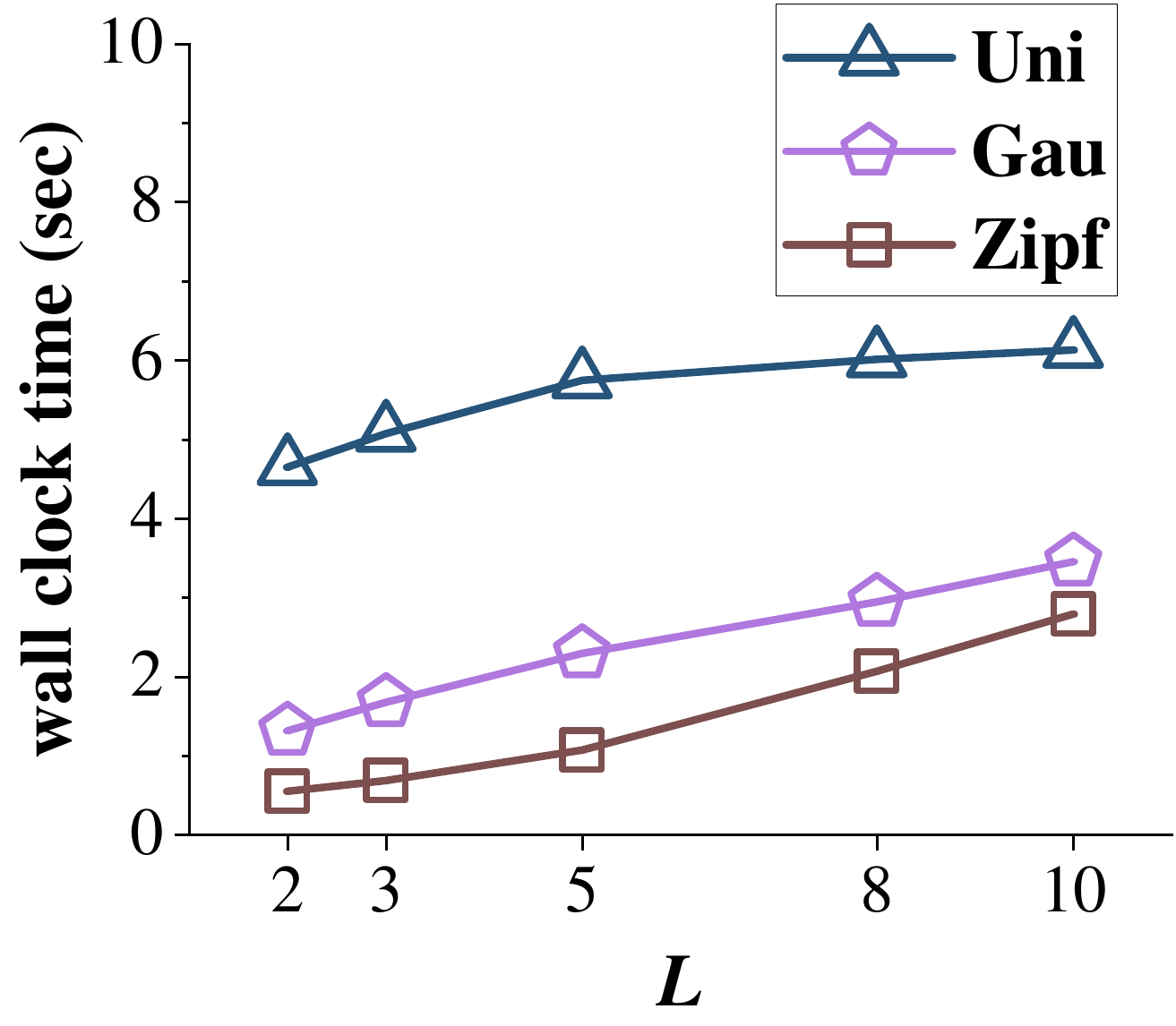}
        \label{subfig:greedy_wp_effect_result_set_size}
    }\hspace{-0.3cm}
    \subfigure[parameter $n$]{
        \includegraphics[height=2.8cm]{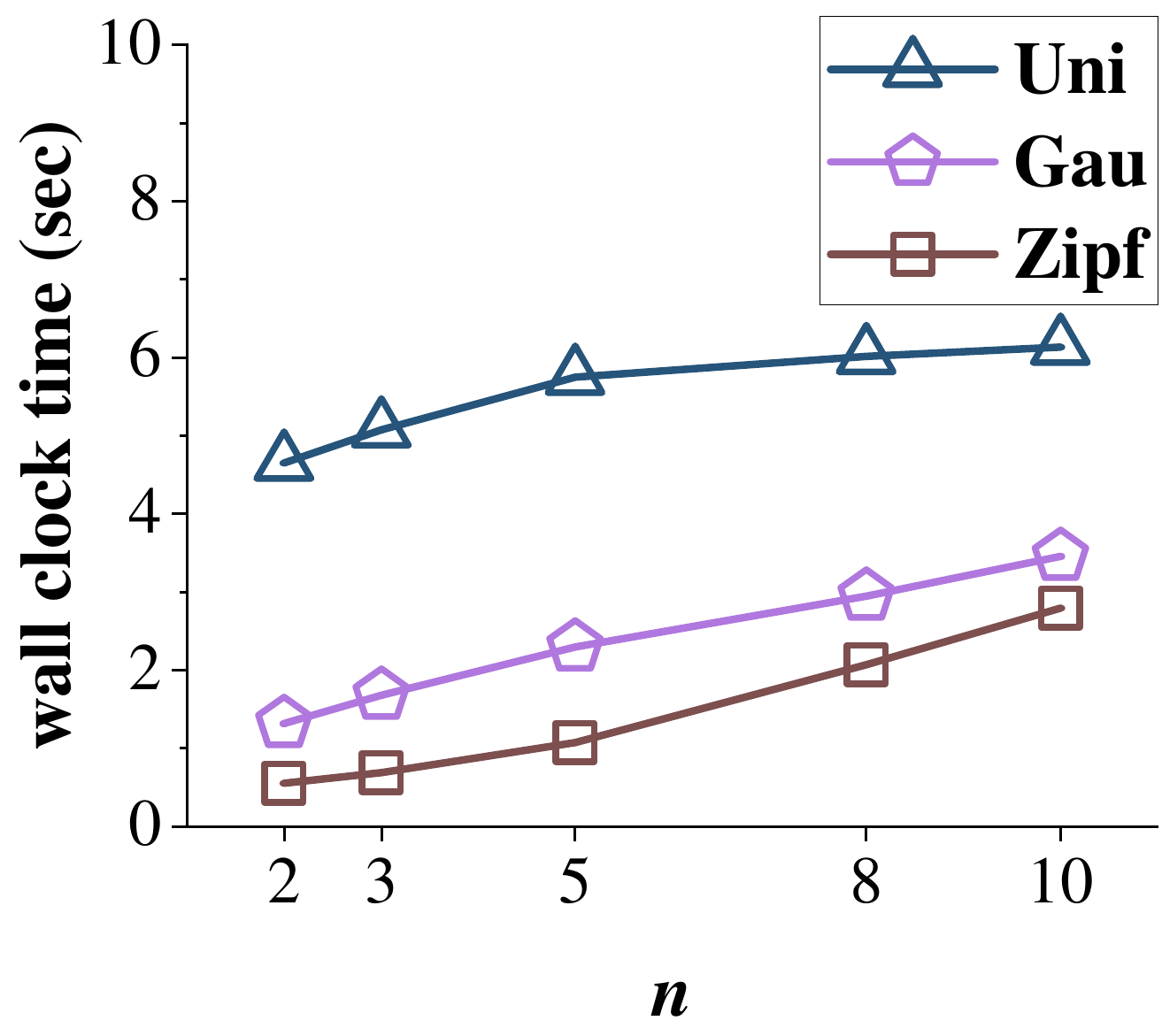}
        \label{subfig:greedy_wp_effect_param_n}
    }\hspace{-0.3cm}
    \subfigure[graph size $|V(G)|$]{
        \includegraphics[height=2.8cm]{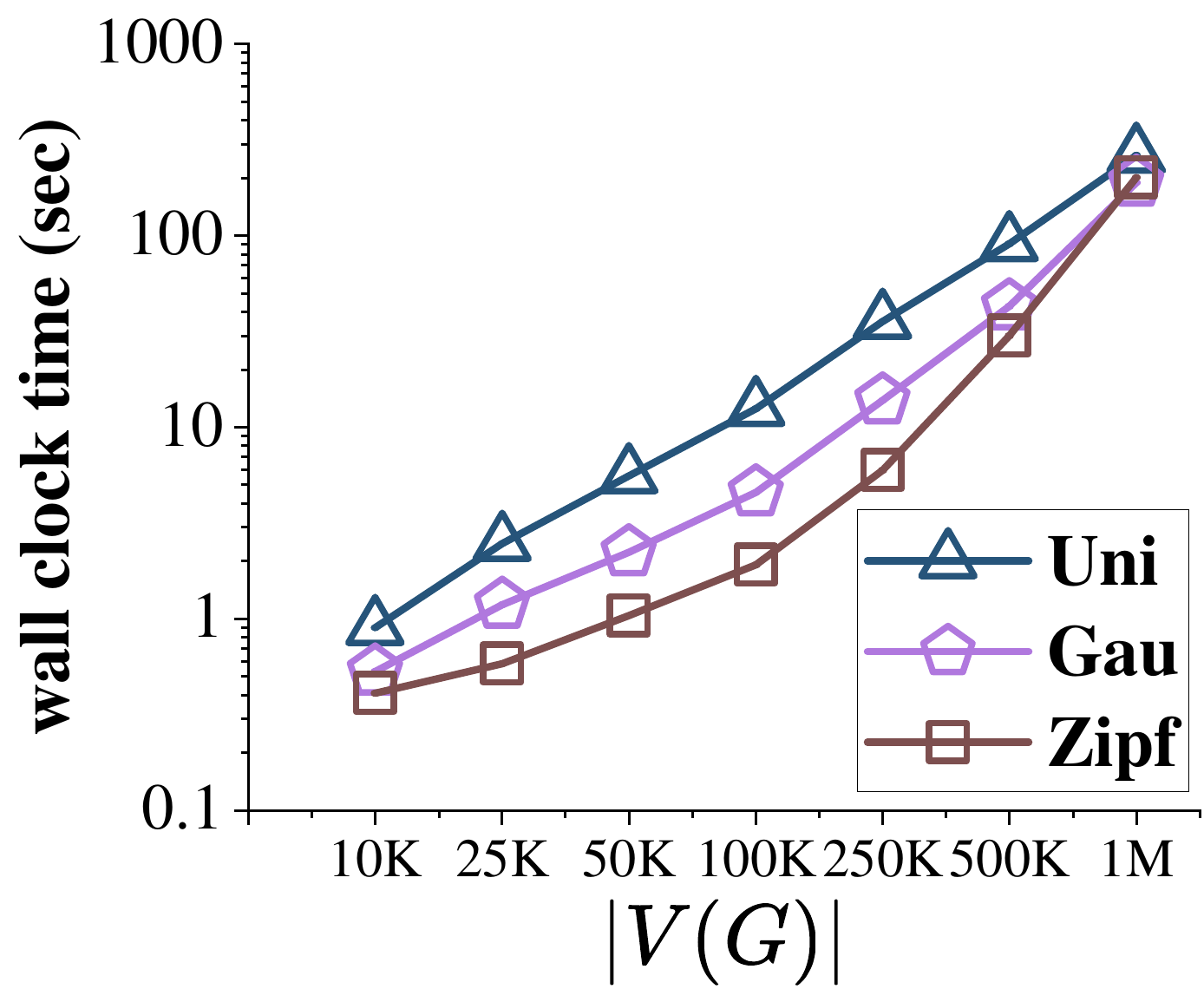}
        \label{subfig:greedy_wp_effect_data_graph_size}
    }\hspace{-0.3cm}
    \subfigure[accuracy vs. graphs]{
        \includegraphics[height=2.8cm]{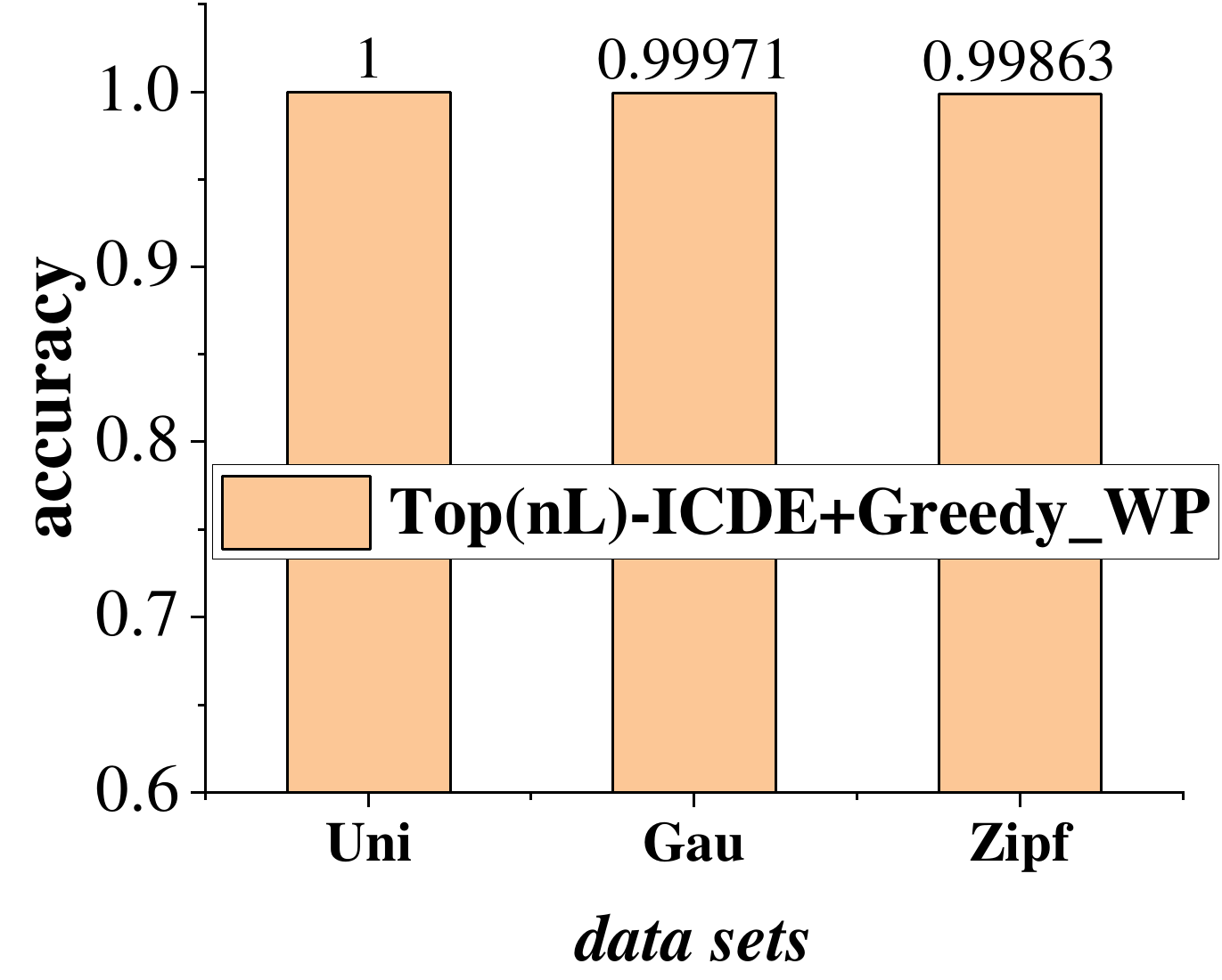}
        \label{subfig:dtopl_performance_accuracy}
    }
    \caption{The DTop$L$-ICDE performance evaluation.}
    \label{fig:DTopL-efficiency}
\end{figure*}


\underline{\it Effect of Keyword Domain Size $|\Sigma|$:}
Figure~\ref{subfig:effect_keywords_domain_size} evaluates the performance of our Top$L$-ICDE approach by varying the keyword domain size, $|\Sigma|$ from $10$ to $80$, where other parameters are set to default values. When $|\Sigma|$ becomes larger, more vertices are expected to be included in seed communities, which results in higher costs to filter/refine larger seed communities. On the other hand, larger seed communities also produce higher influential score bound $\sigma_L$ with higher pruning power. Therefore, for all three graphs, when $|\Sigma|$ increases, the wall clock time first increases and then decreases. Similar to previous experimental results, the time costs remain low (i.e., 1.91 $\sim$ 5.94 $sec$).

\noindent \textbf{Ablation Study (RQ2):} We conduct an ablation study over real/synthetic graphs to evaluate the effectiveness of our proposed pruning strategies, where all parameters are set to their default values. We tested different combinations by adding one more pruning method each time: (1) \textit{keyword pruning} only, (2) \textit{keyword + support pruning}, and (3) \textit{keyword + support + score pruning}. Figure~\ref{subfig:ablation_pruning_power} examines the number of pruned candidate communities, whereas Figure~\ref{subfig:ablation_timecost} shows the time cost for different pruning combinations. From experimental results, we can see that with more pruning methods, the number of pruned communities increases by about an order of magnitude, and the wall clock time decreases. Especially, the third influential score pruning method can significantly prune more candidate communities (in addition to the first two pruning methods) and result in the lowest time cost.

\noindent \textbf{Case Study (RQ3):}
To evaluate the usefulness of our Top$L$-ICDE results, we conduct a case study to compare the influences of our Top$L$-ICDE seed community with that of $k$-core \cite{sozio2010community} over $Amazon$. Figure~\ref{subfig:k_truss_case_study} shows our Top$1$-ICDE community with 4 users ($(4,2)$-truss), whereas Figure~\ref{subfig:k_core_case_study} illustrates 5 users in $4$-core community, where the (red) star point in both subfigures represent the same center vertex. From the figures, we can see that our Top$1$-ICDE community has an influential score $\sigma(g) = 344.31$ with 974 possibly influenced users (blue points). In contrast, the $4$-core has more seed users, but with lower influential score $\sigma(g) = 239.81$ and smaller number of possibly influenced users (i.e., 646). This confirms the usefulness of our Top$L$-ICDE problem to obtain seed communities with high influences for real-world applications such as online advertising/marketing.

\underline{\it Effect of the Graph Size $|V(G)|$:}
Figure~\ref{subfig:effect_data_graph_size} tests the scalability of our Top$L$-ICDE approach with different social network sizes, $|V(G)|$, from $10K$ to $1M$, where default values are assigned to other parameters. In the figure, when the graph size $|V(G)|$ becomes larger, the wall clock time smoothly increases (i.e.,  from 0.51 $sec$ to 255.62 $sec$ for $|V(G)|$ from $10K$ to $1M$, respectively), which confirms the scalability of our Top$L$-ICDE algorithm for large network sizes.

\subsection{DTop$L$-ICDE Performance Evaluation}

\noindent \textbf{The DTop$L$-ICDE Efficiency on Real/Synthetic Graphs (RQ1):}
Figure~\ref{subfig:dtopl_performance_time} compares the performance of our DTop$L$-ICDE approach (i.e., Top($nL$)-ICDE+{\sf Greedy\_WP}), Top($nL$)-ICDE+{\sf Greedy\_WoP}, and the \textit{Optimal} algorithm over real and synthetic graphs, in terms of the \textit{wall clock time}, where all parameters are set to their default values in Table~\ref{tab:parameters}. The experimental results show that our DTop$L$-ICDE approach outperforms \textit{Optimal} by at least three orders of magnitude. 


Below, we evaluate the robustness of our DTop$L$-ICDE approach with different parameters (e.g., $n$, $L$, $|V(G)|$) on synthetic graphs.

\underline{\it Effect of the Size, $L$, of Query Result Set:} 
Figure~\ref{subfig:greedy_wp_effect_result_set_size} illustrates the experimental results of our DTop$L$-ICDE approach for different sizes, $L$, of query result set, by varying $L$ from $2$ to $10$ and default values are used for other parameters. Larger $L$ values lead to lower influential score bound $\sigma_{(nL)}$, and thus more candidate communities to be retrieved and refined, which incur higher time costs. For different $L$ values, the time cost of DTop$L$-ICDE remains low (i.e., 2.72 $\sim$ 6.39 $sec$).


\underline{\it Effect of Parameter $n$:}
Figure~\ref{subfig:greedy_wp_effect_param_n} shows the performance of our DTop$L$-ICDE approach, where $n$ varies from $2$ to $10$ and other parameters are set to their default values. With increasing $n$, lower influential score bound $\sigma_{(nL)}$ is used, resulting in higher time costs. Nevertheless, the \textit{wall clock time} remains small (i.e., around 2.72 $\sim$ 6.28 $sec$) for various $n$ values.


\underline{\it Effect of the Graph Size $|V(G)|$:}
Figure~\ref{subfig:greedy_wp_effect_data_graph_size} reports the performance of our DTop$L$-ICDE approach with different social-network sizes, $|V(G)|$, from $10K$ to $1M$ and default values are used for other parameters. Intuitively, the larger $|V(G)|$, the more communities that must be processed, which incurs smoothly increasing time costs (i.e., 0.9 $\sim$ 278.18 $sec$).

\noindent \textbf{The DTop$L$-ICDE Accuracy (RQ4):}
Moreover, we also test the experiments on small-scale graphs ($|V(G)| = 1K$, $3$ keywords per vertices, and $|\Sigma| = 20$) following Uniform, Gaussian, and Zipf distributions, and report the \textit{accuracy} of our DTop$L$-ICDE approach (i.e., the diversity score ratio of our approach to \textit{Optimal}) in Figure~\ref{subfig:dtopl_performance_accuracy}. The results indicate that our DTop$L$-ICDE accuracy is very close to 100\% (i.e., 99.863\% $\sim$ 100\%).


\section{Related Work}
\label{sec:related_work}


\noindent {\bf Community Search/Detection:}
In real applications of social network analysis, the \textit{community search} (CS) and \textit{community detection} (CD) are two fundamental tasks, which retrieve \textit{one} community that contains a given query vertex and \textit{all} communities in social networks, respectively. 


Prior works proposed many community semantics based on different structural cohesiveness, such as the minimum degree \cite{cui2014local}, $k$-core \cite{sozio2010community}, $k$-clique \cite{cui2013online}, and $(k,d)$-truss \cite{huang2017attribute}. In contrast, our Top$L$-ICDE problem retrieves not only highly connected seed communities but also those with the highest influences and containing query keywords in social networks, which is more challenging. On the other hand, previous works on community detection retrieved all communities by considering link information only \cite{newman2004finding,fortunato2010community}. More recent works used clustering techniques to detect communities \cite{xu2012model,conte2018d2k,veldt2018correlation}. However, these works did not require structural constraints of community answers or consider the impact of the influenced communities, which is the focus of our Top$L$-ICDE problem in this paper.

\noindent {\bf Influence Maximization:} 
Previous works on the \textit{influence maximization} (IM) problem over social networks \cite{chen2010ScalableInfluenceMaximization,tang2018online,ohsaka2020solution,ali2021fairness,li2021large} usually obtain arbitrary (connected or disconnected) individual users from social networks with the maximum influence on other users, where \textit{independent cascade} (IC) and \textit{linear threshold} (LT) models \cite{kempe2003MaximizingSpreadInfluence} were used to capture the influence propagation. However, most existing works on the IM problem do not assume strong social relationships of the selected users (i.e., no community requirements for users). In contrast, our Top$L$-ICDE requires seed communities to be connected, have high structural cohesiveness, and cover some query keywords, which is more challenging.  


\noindent{\bf Influential Community:} 
There are some recent works \cite{li2017most,luo2023EfficientInfluentialCommunity,zhou2023influential} on finding the most influential community over social networks. These works considered different graph data models such as uncertain graphs \cite{luo2023EfficientInfluentialCommunity} and heterogeneous information networks \cite{zhou2023influential} and influential community semantics like $kr$-clique \cite{li2017most}, $(k,\eta)$-influential community, and $(k,\mathcal{P})$-core \cite{zhou2023influential}. Moreover, they did not take into account users' interests (represented by keywords) in communities. With different graph models and influential community semantics, our Top$L$-ICDE problem uses a certain, directed graph data model with the MIA model for the influence propagation and aims to retrieve top-$L$ influential communities (rather than all communities) under different community semantics of structural, keyword-aware, and influential score ranking. Thus, we cannot directly borrow previous works to solve our Top$L$-ICDE problem.









\noindent{\bf Diversified Subgraphs:} There are several existing works that considered retrieving the diversified subgraphs. For example, Yang et al. \cite{yang2016DiversifiedTopkSubgraph} studied the top-$k$ diversified subgraph problem, which returns a set of up to $k$ subgraphs that are isomorphic to a given query graph, and cover the largest number of vertices. Some prior works \cite{huangTopKStructuralDiversity,tai2014StructuralDiversityResisting} studied the structural diversity search problem in graphs, which obtains vertex(es) with the highest structural diversities (defined as $\#$ of connected components in the 1-hop subgraph of a vertex). Chowdhary et al. \cite{chowdhary2022FindingAttributeDiversified} aimed to search for a community that is structure-cohesive (i.e., with the minimum number of vertices) and attribute-diversified (i.e., with the maximum number of attribute labels in vertices). The aforementioned works either did not consider the cohesiveness and/or influences of the returned subgraphs, or focused on node-/attribute-level diversity (rather than community-level diversity). Thus, with different problem definitions, we cannot directly use techniques proposed in these works to solve our DTop$L$-ICDE problem.

\balance

\section{Conclusions}
\label{sec:conclusion}
In this paper, we proposed
a novel Top$L$-ICDE problem, which retrieves top-$L$ seed communities from social networks with the highest influential scores. Unlike traditional community detection problems, the Top$L$-ICDE problem considers influence effects from seed communities, rather than returning highly connected communities alone. To process the Top$L$-ICDE operator, we proposed effective community-level and index-level pruning strategies to filter out false alarms of candidate seed communities and design an effective tree index over pre-computed data to facilitate our proposed efficient Top$L$-ICDE processing algorithm. We also formulated and tackled a Top$L$-ICDE variant (i.e., DTop$L$-ICDE), which is NP-hard, by proposing a greedy algorithm with effective diversity score pruning. Comprehensive experimental results
on real/synthetic social networks confirm the effectiveness and efficiency of our Top$L$-ICDE and DTop$L$-ICDE approaches. 

\bibliographystyle{IEEEtran}
\bibliography{cites}

\begin{thebibliography}{10}
\providecommand{\url}[1]{#1}
\csname url@samestyle\endcsname
\providecommand{\newblock}{\relax}
\providecommand{\bibinfo}[2]{#2}
\providecommand{\BIBentrySTDinterwordspacing}{\spaceskip=0pt\relax}
\providecommand{\BIBentryALTinterwordstretchfactor}{4}
\providecommand{\BIBentryALTinterwordspacing}{\spaceskip=\fontdimen2\font plus
\BIBentryALTinterwordstretchfactor\fontdimen3\font minus \fontdimen4\font\relax}
\providecommand{\BIBforeignlanguage}[2]{{%
\expandafter\ifx\csname l@#1\endcsname\relax
\typeout{** WARNING: IEEEtran.bst: No hyphenation pattern has been}%
\typeout{** loaded for the language `#1'. Using the pattern for}%
\typeout{** the default language instead.}%
\else
\language=\csname l@#1\endcsname
\fi
#2}}
\providecommand{\BIBdecl}{\relax}
\BIBdecl

\bibitem{chen2010scalable}
W.~Chen, C.~Wang, and Y.~Wang, ``Scalable influence maximization for prevalent viral marketing in large-scale social networks,'' in \emph{Proceedings of the International Conference on Knowledge Discovery and Data Mining (SIGKDD)}, 2010, pp. 1029--1038.

\bibitem{tang2015influence}
Y.~Tang, Y.~Shi, and X.~Xiao, ``Influence maximization in near-linear time: A martingale approach,'' in \emph{Proceedings of the International Conference on Management of Data (SIGMOD)}, 2015, pp. 1539--1554.

\bibitem{tu2022viral}
S.~Tu and S.~Neumann, ``A viral marketing-based model for opinion dynamics in online social networks,'' in \emph{Proceedings of the ACM Web Conference 2022}, 2022, pp. 1570--1578.

\bibitem{song2022friend}
X.~Song, J.~Lian, H.~Huang, M.~Wu, H.~Jin, and X.~Xie, ``Friend recommendations with self-rescaling graph neural networks,'' in \emph{Proceedings of the 28th ACM SIGKDD Conference on Knowledge Discovery and Data Mining}, 2022, pp. 3909--3919.

\bibitem{fan2018octopus}
J.~Fan, J.~Qiu, Y.~Li, Q.~Meng, D.~Zhang, G.~Li, K.-L. Tan, and X.~Du, ``{{OCTOPUS}}: {{An Online Topic-Aware Influence Analysis System}} for {{Social Networks}},'' in \emph{2018 {{IEEE}} 34th {{International Conference}} on {{Data Engineering}} ({{ICDE}})}, Apr. 2018, pp. 1569--1572.

\bibitem{wang2020efficient}
K.~Wang, S.~Wang, X.~Cao, and L.~Qin, ``Efficient {{Radius-Bounded Community Search}} in {{Geo-Social Networks}},'' \emph{IEEE Transactions on Knowledge and Data Engineering}, vol.~34, no.~9, pp. 4186--4200, Sep. 2022.

\bibitem{liu2020truss}
Q.~Liu, M.~Zhao, X.~Huang, J.~Xu, and Y.~Gao, ``Truss-based {{Community Search}} over {{Large Directed Graphs}},'' in \emph{Proceedings of the 2020 {{ACM SIGMOD International Conference}} on {{Management}} of {{Data}}}, May 2020, pp. 2183--2197.

\bibitem{sun2020index}
L.~Sun, X.~Huang, R.-H. Li, B.~Choi, and J.~Xu, ``Index-based intimate-core community search in large weighted graphs,'' \emph{IEEE Transactions on Knowledge and Data Engineering}, vol.~34, no.~9, pp. 4313--4327, 2020.

\bibitem{liu2021efficient}
B.~Liu, F.~Zhang, W.~Zhang, X.~Lin, and Y.~Zhang, ``Efficient community search with size constraint,'' in \emph{2021 IEEE 37th International Conference on Data Engineering (ICDE)}.\hskip 1em plus 0.5em minus 0.4em\relax IEEE, 2021, pp. 97--108.

\bibitem{wang2017game}
C.-Y. Wang, Y.~Chen, and K.~R. Liu, ``Game-theoretic cross social media analytic: How yelp ratings affect deal selection on groupon?'' \emph{IEEE Transactions on Knowledge and Data Engineering}, vol.~30, no.~5, pp. 908--921, 2017.

\bibitem{huang2017AttributeDrivenCommunitySearch}
X.~Huang and L.~V.~S. Lakshmanan, ``Attribute-{{Driven Community Search}},'' \emph{VLDB-2017}, vol.~10, no.~9, pp. 949--960, 2017.

\bibitem{al2020topic}
A.~{Al-Baghdadi} and X.~Lian, ``Topic-based community search over spatial-social networks,'' \emph{Proceedings of the VLDB Endowment}, vol.~13, no.~12, pp. 2104--2117, Aug. 2020.

\bibitem{chen2010ScalableInfluenceMaximization}
W.~Chen, C.~Wang, and Y.~Wang, ``Scalable influence maximization for prevalent viral marketing in large-scale social networks,'' in \emph{Proceedings of the International Conference on Knowledge Discovery and Data Mining (SIGKDD)}, ser. {{KDD}} '10, Jul. 2010, pp. 1029--1038.

\bibitem{kempe2003MaximizingSpreadInfluence}
D.~Kempe, J.~Kleinberg, and {\'E}.~Tardos, ``Maximizing the spread of influence through a social network,'' in \emph{Proceedings of the Ninth {{ACM SIGKDD}} International Conference on {{Knowledge}} Discovery and Data Mining}, Aug. 2003, pp. 137--146.

\bibitem{feige1998ThresholdLnApproximating}
U.~Feige, ``A threshold of ln n for approximating set cover,'' \emph{Journal of the ACM}, vol.~45, no.~4, pp. 634--652, Jul. 1998.

\bibitem{cohen2008TrussesCohesiveSubgraphs}
J.~Cohen, ``Trusses: {{Cohesive}} subgraphs for social network analysis,'' \emph{National security agency technical report}, vol.~16, no. 3.1, 2008.

\bibitem{feige1996ThresholdLnApproximating}
U.~Feige, ``A threshold of ln n for approximating set cover (preliminary version),'' in \emph{{{STOC-96}}}, ser. {{STOC}} '96, Jul. 1996, pp. 314--318.

\bibitem{luo2023EfficientInfluentialCommunity}
W.~Luo, X.~Zhou, K.~Li, Y.~Gao, and K.~Li, ``Efficient {{Influential Community Search}} in {{Large Uncertain Graphs}},'' \emph{IEEE Transactions on Knowledge and Data Engineering}, vol.~35, no.~4, pp. 3779--3793, Apr. 2023.

\bibitem{zhou2023InfluentialCommunitySearch}
Y.~Zhou, Y.~Fang, W.~Luo, and Y.~Ye, ``Influential {{Community Search}} over {{Large Heterogeneous Information Networks}},'' \emph{Proceedings of the VLDB Endowment}, vol.~16, no.~8, pp. 2047--2060, Apr. 2023.

\bibitem{chen2015OnlineTopicawareInfluence}
S.~Chen, J.~Fan, G.~Li, J.~Feng, K.-l. Tan, and J.~Tang, ``Online topic-aware influence maximization,'' \emph{Proceedings of the VLDB Endowment}, vol.~8, no.~6, pp. 666--677, Feb. 2015.

\bibitem{newman1999RenormalizationGroupAnalysis}
M.~E.~J. Newman and D.~J. Watts, ``Renormalization group analysis of the small-world network model,'' \emph{Physics Letters A}, vol. 263, no.~4, pp. 341--346, Dec. 1999.

\bibitem{huang2017attribute}
X.~Huang and L.~V.~S. Lakshmanan, ``Attribute-driven community search,'' \emph{Proceedings of the VLDB Endowment}, vol.~10, no.~9, pp. 949--960, May 2017.

\bibitem{sozio2010community}
M.~Sozio and A.~Gionis, ``The community-search problem and how to plan a successful cocktail party,'' in \emph{Proceedings of the 16th {{ACM SIGKDD}} International Conference on {{Knowledge}} Discovery and Data Mining}, ser. {{KDD}} '10, Jul. 2010, pp. 939--948.

\bibitem{cui2014local}
W.~Cui, Y.~Xiao, H.~Wang, and W.~Wang, ``Local search of communities in large graphs,'' in \emph{Proceedings of the 2014 {{ACM SIGMOD International Conference}} on {{Management}} of {{Data}}}, ser. {{SIGMOD}} '14, Jun. 2014, pp. 991--1002.

\bibitem{cui2013online}
W.~Cui, Y.~Xiao, H.~Wang, Y.~Lu, and W.~Wang, ``Online search of overlapping communities,'' in \emph{Proceedings of the 2013 {{ACM SIGMOD International Conference}} on {{Management}} of {{Data}}}, ser. {{SIGMOD}} '13, Jun. 2013, pp. 277--288.

\bibitem{newman2004finding}
M.~E. Newman and M.~Girvan, ``Finding and evaluating community structure in networks,'' \emph{Physical Review E}, vol.~69, no.~2, p. 026113, 2004.

\bibitem{fortunato2010community}
S.~Fortunato, ``Community detection in graphs,'' \emph{Physics reports}, vol. 486, no. 3-5, pp. 75--174, 2010.

\bibitem{xu2012model}
Z.~Xu, Y.~Ke, Y.~Wang, H.~Cheng, and J.~Cheng, ``A model-based approach to attributed graph clustering,'' in \emph{Proceedings of International Conference on Management of Data (SIGMOD)}, 2012, pp. 505--516.

\bibitem{conte2018d2k}
A.~Conte, T.~De~Matteis, D.~De~Sensi, R.~Grossi, A.~Marino, and L.~Versari, ``D2k: scalable community detection in massive networks via small-diameter k-plexes,'' in \emph{Proceedings of ACM SIGKDD Conference on Knowledge Discovery and Data Mining (KDD)}, 2018, pp. 1272--1281.

\bibitem{veldt2018correlation}
N.~Veldt, D.~F. Gleich, and A.~Wirth, ``A {{Correlation Clustering Framework}} for {{Community Detection}},'' in \emph{Proceedings of the 2018 {{World Wide Web Conference}}}, ser. {{WWW}} '18.\hskip 1em plus 0.5em minus 0.4em\relax {Republic and Canton of Geneva, CHE}: {International World Wide Web Conferences Steering Committee}, Apr. 2018, pp. 439--448.

\bibitem{tang2018online}
J.~Tang, X.~Tang, X.~Xiao, and J.~Yuan, ``Online processing algorithms for influence maximization,'' in \emph{Proceedings of the 2018 International Conference on Management of Data}, 2018, pp. 991--1005.

\bibitem{ohsaka2020solution}
N.~Ohsaka, ``The solution distribution of influence maximization: A high-level experimental study on three algorithmic approaches,'' in \emph{Proceedings of the 2020 ACM SIGMOD international conference on management of data}, 2020, pp. 2151--2166.

\bibitem{ali2021fairness}
J.~Ali, M.~Babaei, A.~Chakraborty, B.~Mirzasoleiman, K.~Gummadi, and A.~Singla, ``On the fairness of time-critical influence maximization in social networks,'' \emph{IEEE Transactions on Knowledge and Data Engineering}, 2021.

\bibitem{li2021large}
D.~Li, J.~Liu, J.~Jeon, S.~Hong, T.~Le, D.~Lee, and N.~Park, ``Large-scale data-driven airline market influence maximization,'' in \emph{Proceedings of ACM SIGKDD Conference on Knowledge Discovery and Data Mining (KDD)}, 2021, pp. 914--924.

\bibitem{li2017most}
J.~Li, X.~Wang, K.~Deng, X.~Yang, T.~Sellis, and J.~X. Yu, ``Most influential community search over large social networks,'' in \emph{2017 IEEE 33rd international conference on data engineering (ICDE)}.\hskip 1em plus 0.5em minus 0.4em\relax IEEE, 2017, pp. 871--882.

\bibitem{zhou2023influential}
Y.~Zhou, Y.~Fang, W.~Luo, and Y.~Ye, ``Influential {{Community Search}} over {{Large Heterogeneous Information Networks}},'' \emph{Proceedings of the VLDB Endowment}, vol.~16, no.~8, pp. 2047--2060, Apr. 2023.

\bibitem{yang2016DiversifiedTopkSubgraph}
Z.~Yang, A.~W.-C. Fu, and R.~Liu, ``Diversified {{Top-k Subgraph Querying}} in a {{Large Graph}},'' in \emph{Proceedings of the International Conference on Management of Data (SIGMOD)}, ser. {{SIGMOD}} '16, Jun. 2016, pp. 1167--1182.

\bibitem{huangTopKStructuralDiversity}
X.~Huang, H.~Cheng, R.-H. Li, L.~Qin, and J.~X. Yu, ``Top-{{K}} structural diversity search in large networks,'' \emph{The VLDB Journal}, vol.~24, no.~3, pp. 319--343, Jun. 2015.

\bibitem{tai2014StructuralDiversityResisting}
C.-H. Tai, P.~S. Yu, D.-N. Yang, and M.-S. Chen, ``Structural {{Diversity}} for {{Resisting Community Identification}} in {{Published Social Networks}},'' \emph{IEEE Transactions on Knowledge and Data Engineering}, vol.~26, no.~1, pp. 235--252, Jan. 2014.

\bibitem{chowdhary2022FindingAttributeDiversified}
A.~A. Chowdhary, C.~Liu, L.~Chen, R.~Zhou, and Y.~Yang, ``Finding attribute diversified community over large attributed networks,'' \emph{World Wide Web}, vol.~25, no.~2, pp. 569--607, Mar. 2022.

\end{thebibliography}

\end{document}